\documentclass[nopacs,nokeys,11pt,notitlepage]{revtex4}

 \usepackage{hyperref}
\hypersetup{colorlinks=true,citecolor=blue,linkcolor=blue,filecolor=blue,urlcolor=blue,breaklinks=true}
\usepackage{amsmath}
\usepackage[marginal]{footmisc}
\usepackage{url}
\usepackage{theorem}
\newtheorem{definition}{Definition}

\newtheorem{lemma}[definition]{Lemma}

\newtheorem{observation}[definition]{Observation}

\newtheorem{theorem}[definition]{Theorem}
\newtheorem{corollary}[definition]{Corollary}

\newtheorem{problem}[definition]{Problem}
\def\squareforqed{\hbox{\rlap{$\sqcap$}$\sqcup$}}
\def\qed{\ifmmode\squareforqed\else{\unskip\nobreak\hfil
\penalty50\hskip1em\null\nobreak\hfil\squareforqed
\parfillskip=0pt\finalhyphendemerits=0\endgraf}\fi}
\def\endenv{\ifmmode\;\else{\unskip\nobreak\hfil
\penalty50\hskip1em\null\nobreak\hfil\;
\parfillskip=0pt\finalhyphendemerits=0\endgraf}\fi}
\newenvironment{proof}{\noindent \textbf{{Proof~} }}{\qed}
\newenvironment{remark}{\noindent \textbf{{Remark~}}}{}

\usepackage{graphicx,,epsfig,amsmath,latexsym,amssymb,verbatim,color}
\usepackage{dsfont}
\usepackage{theorem}\usepackage{url}

\usepackage[T1]{fontenc}
\usepackage{lmodern}
\usepackage{bbm}
\usepackage{multirow,bigdelim}
\usepackage{chngcntr}
\usepackage{braket}
\usepackage{extarrows}
\usepackage{theorem}
\usepackage{multirow,bigdelim}

\def\qed{\ifmmode\squareforqed\else{\unskip\nobreak\hfil
\penalty50\hskip1em\null\nobreak\hfil\squareforqed
\parfillskip=0pt\finalhyphendemerits=0\endgraf}\fi}
\def\endenv{\ifmmode\;\else{\unskip\nobreak\hfil
\penalty50\hskip1em\null\nobreak\hfil\;
\parfillskip=0pt\finalhyphendemerits=0\endgraf}\fi}

\mathchardef\ordinarycolon\mathcode`\:
\mathcode`\:=\string"8000
\def\vcentcolon{\mathrel{\mathop\ordinarycolon}}
\begingroup \catcode`\:=\active
  \lowercase{\endgroup
  \let :\vcentcolon
  }

\newcommand{\diad}[2]{\left|#1\right>\left<#2\right|}
\newcommand{\proj}[1]{\diad{#1}{#1}}

\newcommand{\nc}{\newcommand}
\newcommand\C{\mathbbm{C}}

\newcommand{\A}{\mathcal{A}}
\newcommand{\SL}{\mathrm{SL}}

\newcommand\R{\mathbbm{R}}
\newcommand\N{\mathbbm{N}}

\newcommand\F{\mathbb{F}}

\newcommand{\Hm}{\mathcal{H}}

\def\Tr{{\rm Tr}} \def\tr{{\Tr}}

\nc{\Rank}{\operatorname{Rank}}
\nc{\smfrac}[2]{\mbox{$\frac{#1}{#2}$}}

\nc{\ox}{\otimes}
\nc{\dg}{\dagger}
\nc{\dn}{\downarrow}
\nc{\cA}{{\cal A}}
\nc{\cB}{{\cal B}}
\nc{\cC}{{\cal C}}
\nc{\cD}{{\cal D}}
\nc{\cE}{{\cal E}}
\nc{\cF}{{\cal F}}
\nc{\cG}{{\cal G}}
\nc{\cH}{{\cal H}}
\nc{\cI}{{\cal I}}
\nc{\cJ}{{\cal J}}
\nc{\cK}{{\cal K}}
\nc{\cL}{{\cal L}}
\nc{\cM}{{\cal M}}
\nc{\cN}{{\cal N}}
\nc{\cO}{{\cal O}}
\nc{\cP}{{\cal P}}
\nc{\cQ}{{\cal Q}}
\nc{\cR}{{\cal R}}
\nc{\cS}{{\cal S}}
\nc{\cT}{{\cal T}}
\nc{\cX}{{\cal X}}
\nc{\cY}{{\cal Y}}
\nc{\cZ}{{\cal Z}}

\nc{\V}{\text{vec}}
\nc{\rank}{{rank}}
\nc{\sch}{\text{sch}}
\nc{\supp}{\rm supp}
\nc{\spa}{\rm span}

\nc{\be}{\beta}
\nc{\SLOCCTO}{\overset{\underset{\mathrm{SLOCC}}{}}{\longrightarrow}}

\newcommand{\im}{{\rm Im}}
\newcommand{\kernel}{{\rm Ker}}
\newcommand{\dims}{{\rm dim}}
\newcommand{\rk}{mrk}

\begin{document}

\title{Tripartite-to-bipartite Entanglement Transformation by Stochastic Local
Operations and Classical Communication and the Structure of Matrix Spaces
}

\author{Yinan Li $^{1}$}
\email{Yinan.Li@student.uts.edu.au}
\author{Youming Qiao$^{1}$}
\email{Youming.Qiao@uts.edu.au}
\author{Xin Wang$^{1}$}
\email{Xin.Wang-8@student.uts.edu.au}
\author{Runyao Duan$^{1,2}$}
\email{Runyao.Duan@uts.edu.au}
\affiliation{$^1$Centre for Quantum Software and Information, Faculty of Engineering and Information Technology, University of Technology Sydney, NSW 2007, Australia}
\affiliation{$^2$UTS-AMSS Joint Research Laboratory for Quantum Computation and Quantum Information Processing, Academy of Mathematics and Systems Science, Chinese Academy of Sciences, Beijing 100190, China}





\begin{abstract}
We study the problem of transforming a tripartite pure state to a bipartite one using stochastic local operations and classical 
communication (SLOCC). It is known that the tripartite-to-bipartite SLOCC 
convertibility is characterized by the \emph{maximal Schmidt rank} of the given 
tripartite state, i.e.~the largest Schmidt rank over those bipartite states lying 
in the support of the reduced density operator.  In this paper, we further study
this problem and exhibit novel results in both multi-copy and asymptotic settings, 
about properties of the maximal Schmidt rank, utilizing powerful results from 
the structure of matrix spaces.

In the multi-copy regime, we observe that the maximal Schmidt rank 
is strictly super-multiplicative, i.e.~the maximal Schmidt rank of the tensor 
product of two tripartite pure states can be strictly larger than the product 
of their maximal Schmidt ranks. We then provide a full 
characterization of those tripartite 
states whose maximal Schmidt rank is strictly super-multiplicative 
when taking tensor product with itself. Notice that such tripartite states admit strict advantages in tripartite-to-bipartite SLOCC transformation 
when multiple copies are provided. 

In the asymptotic setting,  we focus on determining the tripartite-to-bipartite SLOCC entanglement 
transformation rate. Computing this rate turns out to be equivalent 
to computing the \emph{asymptotic maximal Schmidt rank} of the tripartite state, defined as 
the regularization of its maximal Schmidt rank.
Despite the difficulty caused by the super-multiplicative property, 
we provide explicit formulas for evaluating the asymptotic maximal 
Schmidt ranks of two important families of 
tripartite pure states, by resorting to certain results of the 
structure of matrix spaces, including the study of matrix semi-invariants. These formulas turn out to be powerful enough to give 
a sufficient and necessary condition to determine whether a given tripartite pure state can be transformed to the bipartite maximally entangled state under SLOCC, in the asymptotic setting. Applying the recent progress on the non-commutative rank 
problem, we can verify this condition in deterministic polynomial time.
\end{abstract}

\maketitle
\section{Introduction}
As a key concept in quantum mechanics, entanglement plays a central role in 
quantum information processing. It is the resource responsible for the quantum 
computational speed-up, quantum communication, quantum cryptography and so on.  In the well-known bipartite case, there is no doubt that the bipartite maximally entangled state plays an important role in quantum 
information theory, since it is usually sufficient to perform many quantum 
information processing tasks, such as quantum teleportation~\cite{Bennett1993} and 
superdense coding~\cite{Bennett1992}. 
Unfortunately, in practice, it becomes very difficult to preserve the bipartite pure entanglement as the two systems may interact with other 
systems (e.g.~the environment), thus changing the situation to the multipartite setting. Consequently, a 
very natural question to ask is, how many bipartite pure entangled states can be 
distilled from a multipartite state, by means of \emph{local operations and 
classical communication} (LOCC)? Previous works on this problem 
introduce various concepts such as the entanglement of 
assistance~\cite{cohen1998,DiVincenzo1999,Smolin2005}, the localizable 
entanglement~\cite{PhysRevLett.92.027901,PhysRevA.71.042306,horodecki2005partial}, 
the concurrence of assistance~\cite{PhysRevA.72.042329}, the random state 
entanglement~\cite{PhysRevLett.98.260501}, the entanglement of 
collaboration~\cite{PhysRevA.73.062331,gour2006entanglement} and the entanglement of 
combing~\cite{PhysRevLett.103.220501}. These concepts have been shown vastly useful in other 
areas of quantum information theory, including the study of 
environment-assisted capacity of quantum channels~\cite{winter2005a}, unital 
quantum channels, and the quantum Birkhoff's theorem~\cite{Mendl2009}. 

A slightly different setting is also of great interest. 
Roughly speaking, assume parties A, B and C share a tripartite pure state. Their goal 
is to recover some bipartite pure entanglement with a nonzero probability by LOCC, with the help of C.  
Such protocols usually refer to \emph{stochastic local operations and classical communication} (SLOCC), 
which has been widely used to study entanglement 
classification~\cite{PhysRevA.62.062314,PhysRevLett.111.060502} and entanglement 
transformation~\cite{Chitambar2008,Yu2010,PhysRevLett.105.200501}. 
The advantage of using SLOCC over LOCC is that SLOCC operations admit a simpler
mathematical structure~\cite{Dur2000}. It is known that the SLOCC bipartite 
entanglement transformation can be simply characterized by the \emph{Schmidt rank}~\cite{Vidal1999,lo2001}.  The Schmidt rank of a bipartite state $\ket{\phi}$ is the minimum number of product states needed to express it and is denoted by $srk(\ket{\phi})$. Then $\ket{\phi}$ can be transformed to another bipartite state $\ket{\psi}$, symbolically expressed as $\ket{\phi}\SLOCCTO\ket{\psi}$, 
if and only if $srk(\ket{\phi})\geq srk(\ket{\psi})$. For SLOCC tripartite 
entanglement transformations, the \emph{tensor ranks} of tripartite states are nonincreasing under SLOCC, providing an important entanglement monotone~\cite{Chitambar2008}. 

In these contexts, besides characterizing the feasibility of transformations, one may also consider such problems from the algorithmic viewpoint. One important problem is to find efficient algorithms, which, given (the classical description of) two multipartite states, 
decides one can be transformed to the other using SLOCC. 
This perspective can be made precise via computational 
complexity theory. It is well-known that computing the Schmidt rank of a 
bipartite state (equivalent to computing the rank of a matrix) admits a 
deterministic polynomial-time algorithm, which can be used to determine the 
bipartite SLOCC convertibility. On the other hand, computing the tensor rank of a 
tripartite state is NP-hard~\cite{Hastad1990}, and based on this, 
Chitambar, Duan and Shi~\cite{Chitambar2008} have shown that deciding the SLOCC 
convertibility for tripartite states is also NP-hard in general.

Going back to the multipartite-to-bipartite cases, Chitambar, Duan and Shi~\cite{Chitambar2010} have shown that deciding the multipartite-to-bipartite SLOCC convertibility is equivalent to the \emph{polynomial identity testing} (PIT) problem, which is one of the most important problems in theoretical computer science with many applications, such as in perfect matching~\cite{edmonds1967systems}, multiset equality testing~\cite{Blum1995} and primality testing~\cite{Agrawal2003}. Utilizing the Schwartz-Zippel lemma~\cite{Schwartz1980,Zippel1979}, PIT admits a polynomial-time \emph{randomized} algorithm. However, whether a \emph{deterministic} polynomial-time algorithm for PIT exists is still open. Restricting to the tripartite-to-bipartite SLOCC convertibility, it is equivalent to the \emph{symbolic determinant identity testing} (SDIT) problem~\cite{edmonds1967systems,Gurvits:2004:CCQ:1039323.1039332}, which asks to compute the rank of a given matrix with entries being linear forms over the complex field and is equivalent to PIT for \emph{weakly-skew arithmetic circuits} \cite{Toda92}. The equivalence can be seen as follows: The tripartite-to-bipartite SLOCC convertibility can be determined by the so called \emph{maximal Schmidt rank}~\cite{Chitambar2010}, denoted by $msrk(\cdot)$, which is the highest Schmidt rank over those bipartite states lying in the support of the reduced density operator $\tr_{\rm C}(\proj{{\rm \Psi}_{\rm ABC}})$ shared by A and B. A tripartite state $\ket{{\rm \Psi}_{\rm ABC}}$ can be transformed to a bipartite state $\ket{{\rm \phi}_{\rm AB}}$, if and only if the maximal Schmidt rank of $\ket{{\rm \Psi}_{\rm ABC}}$ is at least the Schmidt rank of $\ket{\phi_{\rm AB}}$. Computing the maximal Schmidt rank is exactly equivalent to computing the rank of a given matrix with entries being linear forms over the complex field. Due to this connection, determine the tripartite-to-bipartite SLOCC convertibility can be solved by polynomial-time randomized algorithms for SDIT (e.g.~\cite{lovasz1979determinants}). Interestingly, under plausible computational assumptions, SDIT must also admit a deterministic polynomial-time algorithm~\cite{nisan1988hardness}. It is believed that to devise such an algorithm would be difficult, as it implies strong circuit lower bounds which seem beyond the current techniques \cite{ik2004}.

From the information-theoretic perspective, it is natural to consider asymptotic tripartite-to-bipartite SLOCC transformations. 
Given $n$ copies of $\ket{{\rm \Psi}_{\rm ABC}}$, let $m(n)$ be the maximum number of copies of 
$\ket{\psi_{\rm AB}}$ which can be obtained by SLOCC. Then in the asymptotic setting, we are interested in computing the ratio $m(n)/n$ as $n$ 
goes to infinity, denoted by $R(\ket{{\rm \Psi}_{\rm ABC}},\ket{\psi_{\rm AB}})$, which is known as the \textit{SLOCC 
entanglement transformation rate} (e.g. see Ref.~\cite{Yu2014a,vrana2015asymptotic}).
By defining the \emph{asymptotic maximal Schmidt rank} of a tripartite state as 
the regularization of its maximal Schmidt rank, the SLOCC transformation rate 
equals the logarithm of the asymptotic maximal Schmidt rank of the given 
tripartite state (where the base of the logarithm is the Schmidt rank of the given 
bipartite state).
Asymptotic SLOCC transformations have also been considered in
the bipartite and tripartite settings, which lead to the concepts of 
the asymptotic Schmidt rank and the asymptotic 
tensor rank, respectively. 
The asymptotic Schmidt rank equals Schmidt rank itself, as the Schmidt rank is multiplicative, i.e.~for bipartite states $\ket{\psi}$ and 
$\ket{\phi}$, the Schmidt rank of $\ket{\psi}\ox\ket{\phi}$ equals the product of the Schmidt ranks of 
$\ket{\psi}$ and $\ket{\phi}$. On the other hand, the tensor rank 
is not multiplicative~\cite{Chitambar2008}, which makes the asymptotic tensor rank notoriously difficult to evaluate. 

In this paper, we systematically study the tripartite-to-bipartite SLOCC entanglement transformations, in both multi-copy and asymptotic setting.
We first illustrate the super-multiplicativity of maximal Schmidt rank, by 
constructing a tripartite state $\ket{\rm{\Psi}_{\rm ABC}}$ satisfying 
$msrk(\ket{\rm{\Psi}_{\rm ABC}}^{\ox 2})>msrk(\ket{\rm{\Psi}_{\rm ABC}})^{2}$. 
Then we provide a characterization of those tripartite states whose maximal Schmidt ranks are strictly increasing on average under tensor product. Notice 
that such tripartite states admit strict advantages in tripartite-to-bipartite SLOCC transformation with multiple copies. Interestingly, except for those degenerated cases, this phenomenon holds for all tripartite states of which their maximal Schmidt ranks are not full.
In the asymptotic setting, one of the most interesting questions is deciding whether the $d\otimes d$ maximally 
entangled state $\ket{\rm{\Phi}_{\rm AB}}:=\frac{1}{\sqrt{d}}\sum_{0\leq i\leq 
d-1}\ket{i_{\rm A}}\ket{i_{\rm B}}$ can be obtained from a given tripartite state $\ket{{\rm \Psi}_{\rm ABC}}$
by SLOCC asymptotically, i.e.~$R(\ket{{\rm \Psi}_{\rm ABC}}, \ket{\rm{\Phi}_{\rm AB}})=1$. Guided by the structure theory of matrix 
spaces, we exhibit explicit formulas to compute the asymptotic maximal Schmidt 
ranks of a large family of tripartite states. To obtain one of the formulas, we resort to 
certain results from invariant theory, specifically from matrix semi-invariants. While the 
use of invariant theory in entanglement theory is common, to the best of our knowledge, 
this is the first time that results from matrix semi-invariants are utilized to 
study SLOCC transformations.
Based on these formulas, we settle the question by providing a full 
characterization of those states that can achieve so. Interestingly, this characterization 
is algorithmically effective, i.e.~there exist deterministic 
polynomial-time algorithms to determine whether this condition holds for a given 
tripartite state~\cite{Garg2015,Ivanyos2015a}.

\paragraph*{Organization.}  
In Section~\ref{slocc}, we present preliminaries about SLOCC transformations, and some background knowledge of 
the structure of matrix spaces. In Section~\ref{finite},  we construct tripartite states 
of which their maximal Schmidt ranks are strictly super-multiplicative, and provide a full 
characterization for those tripartite states that satisfy this property. In Section~\ref{proof}, we explicitly compute the asymptotic maximal Schmidt ranks of a large family of tripartite states and exhibit a sufficient and necessary result to 
determine whether a tripartite state can be transformed to the bipartite maximally 
entangled state by SLOCC, in an asymptotic setting. We close in Section~\ref{conclusion} with a brief conclusion.

\section{Preliminaries and backgrounds}\label{slocc}
\subsection{Preliminaries}\label{subsec:prel}
We use $\Hm^{\rm A}_d$, $\Hm^{\rm B}_d$ and $\Hm^{\rm C}_{d}$ to denote 
$d$-dimensional Hilbert spaces (the underlying field is the complex field 
$\C$) associated with parties A, B and C, respectively. When there is no confusion, 
we assume $\Hm^{\rm A}$ and $\Hm^{\rm B}$ have the same dimension ($d$), and use 
$\{\ket{0},\dots,\ket{d-1}\}$ to denote the computational basis of a 
$d$-dimensional Hilbert spaces. 
For any bipartite pure state $\ket{{\rm \psi}_{\rm AB}}$, which is a unit vector 
in $\Hm^{\rm A}\otimes \Hm^{\rm B}$, $srk(\ket{{\rm \psi}_{\rm AB}})$ denotes the Schmidt rank 
of $\ket{{\rm \psi}_{\rm AB}}$, which is the minimal number of product states 
required to linearly span $\ket{{\rm \psi}_{\rm AB}}$. For a tripartite pure state $\ket{{\rm \Phi}_{\rm 
ABC}}\in\Hm^{\rm A}\otimes \Hm^{\rm B}\ox\Hm^{\rm C}$, let $\rho^{{\rm \Phi}}_{\rm 
AB}=\tr_{\rm C} (\proj{{\rm \Phi}_{\rm ABC}})$ be the reduced density operator 
shared by A and B. The mixed state $\rho^{{\rm \Phi}}_{\rm AB}$ admits a representation as 
$\sum_{i=1}^np_i\proj{\psi_i}$, where $\braket{\psi_i|\psi_j}=\delta_{ij}$ and 
$p_i>0$. The ``subnormalized'' eigenstates $\{\ket{\tilde{\psi}_i}=\sqrt{p_i}\ket{\psi_i}\}_{i=1,\dots,n}$ 
span (with respect to complex numbers) the space $\supp(\rho^{{\rm \Phi}}_{\rm AB})$, which is called the 
support of $\rho^{{\rm \Phi}}_{\rm AB}$. The \emph{maximal Schmidt rank} of a 
tripartite pure state $\ket{{\rm \Phi}_{\rm ABC}}$ is defined by 
\begin{equation}\label{eq: maximal schmidt rank}
msrk(\ket{{\rm \Phi}_{\rm ABC}}):=\max\{srk(\ket{{\rm \phi}_{\rm AB}}):\ket{{\rm \phi}_{\rm AB}}\in \supp(\rho^{{\rm \Phi}}_{\rm AB})\}.
\end{equation}
In the rest of this paper we focus on transforming tripartite \emph{pure} states to bipartite \emph{pure} states by SLOCC,
which can be characterized by the following.
\begin{theorem}[Chitambar, Duan and Shi~\cite{Chitambar2010}]\label{thm: SLOCC condition}
$\ket{{\rm \Phi}_{\rm ABC}}$ can be transformed to $\ket{\psi_{\rm AB}}$ by means of SLOCC
if and only if  $msrk(\ket{{\rm \Phi}_{\rm ABC}})\ge srk(\ket{\psi_{\rm AB}})$.
\end{theorem}

The SLOCC protocol for Theorem~\ref{thm: SLOCC condition} as proposed 
in~\cite{Chitambar2010} takes the following form: 
firstly C makes a measurement on his part of $\ket{{\rm \Phi}_{\rm ABC}}$ and 
broadcasting 
the result to A and B; then there exists an SLOCC protocol by which A and B can 
convert their state to $\ket{\psi_{\rm AB}}$.
This ``one-way'' protocol coincide with the one in the entanglement of assistance~\cite{DiVincenzo1999}. It is also natural to consider the protocol in the entanglement of collaboration~\cite{gour2006entanglement}, which allows two-way communications between A and B on one side, and C on the other side, as follows: before C make measurements, A and B can do measurements on their own 
systems, and broadcast their outcomes to C. It is known that in the LOCC setting, 
such two-way communications are necessary for some tripartite-to-bipartite 
transformations to happen with probability $1$~\cite{PhysRevA.73.062331}. On the 
other hand, in the SLOCC setting, Chitambar, Duan and Shi~\cite{Chitambar2010} 
have shown 
that $\ket{{\rm \Phi}_{\rm ABC}}$ can be transformed to $\ket{\psi_{\rm AB}}$ by means of SLOCC if 
and only if it can be done by a ``one-way'' protocol. 

Theorem~\ref{thm: SLOCC condition} settles the finite-copy case, 
but leaves the asymptotic setting open, which is natural and important from the information theoretic perspective. 
In the asymptotic setting, the ability to transform a tripartite pure state $\ket{{\rm \Psi}_{\rm ABC}}$ to a bipartite pure state $\ket{\psi_{\rm AB}}$ is characterized by 
the \textit{SLOCC entanglement transformation rate} (e.g. see Ref.~\cite{Yu2014a,vrana2015asymptotic}), defined as follows:
\begin{equation}\label{eq:transformation rate} 
R(\ket{{\rm \Psi}_{\rm ABC}},\ket{\psi_{\rm AB}}):=\sup_{n\geq 1}\left\{\frac{1}{n}\max\{m:\ket{{\rm \Psi}_{\rm ABC}}^{\ox n}\SLOCCTO\ket{\psi_{\rm AB}}^{\ox m}\}\right\}.
\end{equation}
Notice that
$\max\{m:\ket{{\rm \Psi}_{\rm ABC}}^{\ox n}\SLOCCTO\ket{\psi_{\rm AB}}^{\ox 
m}\}=\lfloor\log_{srk(\ket{\psi_{\rm AB}})}msrk(\ket{{\rm\Psi}_{\rm ABC}}^{\otimes 
n})\rfloor$ for every fixed $n$. 
Define the \emph{asymptotic maximal Schmidt rank} of $\ket{{\rm \Psi}_{\rm ABC}}$ as
\begin{equation}\label{def:asymptotic maximal Schmidt rank}
msrk^\infty(\ket{{\rm \Psi}_{\rm ABC}}):=\sup_{n\geq 1}\sqrt[n]{msrk(\ket{{\rm \Psi}_{\rm ABC}}^{\ox n})}.
\end{equation}
Then the SLOCC entanglement transformation rate of $\ket{{\rm \Psi}_{\rm ABC}}$ and $\ket{\psi_{\rm AB}}$ can be evaluated by
$$R(\ket{{\rm \Psi}_{\rm ABC}},\ket{\psi_{\rm AB}})=\log_{srk(\ket{\psi_{\rm AB}})}msrk^\infty(\ket{{\rm\Psi}_{\rm ABC}}).$$
Essentially, we can replace taking supremum ``$\sup_{n\geq 1}$'' by taking limit ``$\lim_{n\to\infty}$'', as shown 
in the following lemma:
\begin{lemma}\label{thm:replace sup by lim}
$msrk^\infty(\ket{{\rm\Psi}_{\rm ABC}})$ is finite for all $\ket{{\rm\Psi}_{\rm ABC}}\in\Hm^{\rm A}\otimes \Hm^{\rm B}\ox\Hm^{\rm C}$. Moreover, 
\begin{equation}
msrk^\infty(\ket{{\rm\Psi}_{\rm ABC}})=\lim_{n\to\infty}\sqrt[n]{msrk(\ket{{\rm\Psi}_{\rm ABC}}^{\otimes n})}.
\end{equation}
\end{lemma}
\begin{proof}
We shall utilize the following lemma:
\begin{lemma}[Lemma in appendix A of Ref.~\cite{Barnum1998}]
Suppose $c_1,c_2,\dots$ is a nonnegative sequence such that $c_n\leq kn$ for some $k\geq 0$, and $c_m+c_n\leq c_{m+n}$ for all $m$ and $n$. Then $\lim_{n\to\infty}\frac{c_n}{n}$ exists and is finite. 
\end{lemma}
Let $c_n=\log_2 msrk(\ket{{\rm\Psi}_{\rm ABC}}^{\otimes n})$ and $d=\min\{\dim(\Hm^{\rm A}), \dim(\Hm^{\rm B})\}$. Choosing $k=\log_2 d$, it is easy to see $c_n\leq n\log_2 d$, as $msrk(\ket{{\rm\Psi}_{\rm ABC}}^{\otimes n})\leq d^n$ . On the other hand, to prove $c_m+c_n\leq c_{m+n}$, notice that $srk(\ket{\psi})srk(\ket{\phi})= srk(\ket{\psi}\ox \ket{\phi})$ holds for any bipartite state $\ket{\psi}$ and $\ket{\phi}$. Then it is easy to see that $mrk(\ket{{\rm\Psi}_{\rm ABC}}^{\otimes m}\ox \ket{{\rm\Psi}_{\rm ABC}}^{\otimes n})\geq mrk(\ket{{\rm\Psi}_{\rm ABC}}^{\otimes m})mrk(\ket{{\rm\Psi}_{\rm ABC}}^{\otimes n})$, which leads to $c_m+c_n\leq c_{m+n}$. This ensures that $\lim_{n\to\infty}\frac{1}{n}\log_2 msrk(\ket{{\rm\Psi}_{\rm ABC}}^{\otimes n})=\lim_{n\to\infty} \log_2 \sqrt[n]{msrk(\ket{{\rm\Psi}_{\rm ABC}}^{\otimes n}})$ exists and is finite by lemma~\cite{Barnum1998}. On the other hand, by Fekete's Lemma~\cite{Fekete1923}, the condition that $c_m+c_n\leq c_{m+n}$ for all $m$ and $n$ derives that $\sup_{n\geq 1}\frac{c_n}{n}=\lim_{n\to\infty}\frac{c_n}{n}$. This concludes the proof.
\end{proof}

It is clear that computing the maximal and asymptotic maximal Schmidt rank are physically 
worthwhile. From the algorithmic perspective, 
computing the maximal Schmidt rank of a given tripartite state only admits a randomized polynomial-time 
algorithm~\cite{lovasz1979determinants}, since computing the maximal 
Schmidt rank is equivalent to computing the rank of a matrix whose entries are 
linear forms~\cite{edmonds1967systems}.
We explain how this equivalence works as this sets the stage for our work 
as well. Let the space of linear operators 
from $\Hm^{\rm B}_n$ to $\Hm^{\rm A}_m$ be $\cL(\Hm^{\rm B}_n,\Hm^{\rm A}_m)$. 
The linear map ${\rm vec}: \Hm^{\rm A}_m \ox \Hm^{\rm B}_n \to \cL(\Hm^{\rm B}_n, \Hm^{\rm A}_m)$ is defined by 
${\rm vec}(\ket{i} \ox \ket{j})=\ket{i}\bra{j}$, where $\{\ket{i}:i=0,\dots,m-1\}$ and 
$\{\ket{j}:j=0,\dots,n-1\}$ form orthogonal basis of $\Hm^{\rm A}_m$ and $\Hm^{\rm B}_n$, 
respectively. ${\rm vec}$ is a linear isomorphism between $\cL(\Hm^{\rm B}_n,\Hm^{\rm A}_m)$ and 
$\Hm^{\rm A}_m\otimes \Hm^{\rm B}_n$. Given a bipartite state $\ket{\psi_{\rm AB}}$, its Schmidt rank equals the rank of ${\rm vec}(\ket{\psi_{\rm AB}})$.
For a tripartite pure state $\ket{{\rm \Psi}_{\rm ABC}}\in\Hm^{\rm A}_m\otimes\Hm^{\rm B}_n\otimes\Hm^{\rm C}_d$, we define
\begin{equation}\label{eq:msrk}
M({\rm \Psi}_{\rm ABC}):={\rm vec}[\supp({\rho^{{\rm \Psi}}_{\rm AB}})]={\rm span}\{{\rm vec}(\ket{\psi_{\rm AB}}):\ket{\psi_{\rm AB}}\in\supp({\rho^{{\rm \Psi}}_{\rm AB}})\},
\end{equation}
where the linear span is taken over the complex field $\C$. Notice that $M({\rm 
\Psi}_{\rm ABC})$ is a linear space of linear operators over $\C$. Equivalently, after fixing a basis of $\Hm^{\rm A}_m\otimes \Hm^{\rm B}_n$, 
it is a linear space of $m\times n$ matrices over $\C$, which is also called an 
$m\times n$ \emph{matrix space}. Thus, computing the maximal Schmidt rank of the tripartite state 
$\ket{{\rm \Psi}_{\rm ABC}}$ is equivalent to compute the largest rank over 
matrices in the matrix space $M({\rm \Psi}_{\rm ABC})$. 
We define the \emph{maximal rank} and the \emph{asymptotic maximal rank} 
of a $m\times n$ matrix space $\cS$ as
$$mrk(\cS):=\max\{\rank(E):E\in\cS\},~~mrk^{\infty}(\cS):=\sup_{n\geq 1}\sqrt[n]{mrk(\cS^{\ox n})}=\lim_{n\to\infty}\sqrt[n]{mrk(\cS^{\ox n})},$$
where the second equation can be proved using the same argument in 
lemma~\ref{thm:replace sup by lim}. It is straightforward to see that, for 
$\ket{{\rm \Psi}_{\rm ABC}}\in\Hm^{\rm A}_m\otimes\Hm^{\rm B}_n\ox \Hm^{\rm C}_{d}$, we have
$$msrk(\ket{{\rm \Psi}_{\rm ABC}})=mrk(M({\rm \Psi}_{\rm ABC})),~~msrk^\infty(\ket{{\rm \Psi}_{\rm ABC}})=mrk^\infty(M({\rm \Psi}_{\rm ABC})).$$

In the above, we have reformulated the tripartite-to-bipartite SLOCC 
transformation problem, in both the finite-copy setting and the asymptotic setting, 
as computing the maximal ranks and asymptotic ranks of matrix spaces. Therefore, we need to recall 
some basic properties of matrix spaces. More background knowledge will be covered in the 
next subsection. Let the space of all $m\times n$ matrices over the complex field be $M(m\times 
n,\C)$, and let $M(d, \C):=M(d\times d, \C)$. The computational basis of $M(m\times n,\C)$ is denoted by 
$\{\ket{i}\bra{j}:0\leq i\leq m-1, 0\leq j\leq n-1\}$. We use ${\bf 0}_{m\times n}$ to denote the zero matrix in $M(m\times n,\C)$ (or ${\bf 0}$ when there is no confusion), and $I_d$ to denote the identity matrix in $M(d,\C)$. We use $\cS\leq\cS'$ to denote that 
$\cS$ is a subspace of $\cS'$. Two matrix spaces $\cS,\cS'\leq M(m\times n,\C)$ are \emph{equivalent}, if there exist invertible matrices $P\in M(m,\C)$ and $Q\in M(n,\C)$, such that $\cS=P\cS'Q=\{PSQ:S\in\cS'\}$. It is easy to see that equivalent matrix spaces have the same maximal rank.
A matrix space $\cS\leq M(d,\C)$ is \emph{non-singular}, if it contains at least one full-rank matrix. 
Otherwise we say $\cS$ is \emph{singular}. One important structure for matrix 
spaces is the following so-called \emph{shrunk subspace}: 
\begin{definition}\label{shrunk}
Given $\cS\leq M(d,\C)$, a subspace $U\leq\C^d$ is 
called a \emph{shrunk subspace} of $\cS$, if $\dim(U)>\dim(\cS(U))$,  where $\cS(U):={\rm span}\{\cup_{E\in \cS}\{E\ket{u}:\ket{u}\in U\}\}$.
\end{definition}
In fact, this definition is reminiscent of the \emph{shrunk subset} as in the 
famous Hall's marriage theorem~\cite{Hall1935}. Recall that for a bipartite graph 
$G=(L\cup R, E)$ where $|L|=|R|$, Hall's marriage theorem states that $G$ has a 
perfect matching if and only if $G$ does not have shrunk subset, that is a 
subset $S\subseteq L$ such that $|S|>|N(S)|$ where $N(S)$ denotes the set of 
neighbours of $S$. Getting back to the matrix space setting, it is clear that if 
$\cS$ has shrunk subspaces, then $\cS$ must be singular. However, unlike in the 
bipartite graph setting, it is not true 
that any singular matrix space has shrunk subspaces. For instance, the 
$3 \times 3$ skew symmetric matrix space ${\rm span}\{\ket{i}\bra{j}-\ket{j}\bra{i}:0\leq i 
<j\leq 2\}\leq M(3,\C)$ is singular, and has no shrunk subspace. 
Given a $d\times d$ matrix space $\cS$ which has a shrunk subspace $U$, it admits a particular form when we transform it with 
appropriate base changes. In the following, we present this form and introduce an important family of matrix spaces, the \emph{maximal-compression matrix spaces}. 

Suppose $\dim(U)=d-q$ and $\dim(\cS(U))=p$. By definition~\ref{shrunk}, we know that $p+q<d$.
Fix bases for $U$ and $\cS(U)$ so that $U=\mathrm{span}\{\ket{\alpha_{q}},\dots,\ket{\alpha_{d-1}}\}$ 
and $\cS(U)=\mathrm{span}\{\ket{\beta_0},\dots,\ket{\beta_{p-1}}\}$. 
Extend them to full bases of $\C^d$, i.e.~$\C^d=\mathrm{span}\{\ket{\alpha_0},\dots,\ket{\alpha_{q-1}},\ket{\alpha_{q}},\dots,\ket{\alpha_{d-1}}\}=\mathrm{span}\{\ket{\beta_0},\dots,\ket{\beta_{p-1}},\ket{\beta_{p}},\dots,\ket{\beta_{d-1}}\}$. 
Let $Q_1$ and $Q_2$ be invertible matrices transforming the original bases to the above 
two bases, and let $\cS'=Q_1\cS Q_2^{-1}$. It is easy 
to verify that every matrix $E'$ in $\cS'$  is of the following block 
form:
\begin{equation*}
E'=\left[
\begin{array}{c|c}
A_{p\times q} &  B_{p\times (d-q)} \\
\hline
 C_{(d-p)\times q}& D_{(d-p)\times (d-q)}
\end{array}
\right]_{d\times d},
\end{equation*}
where $A_{p\times q}\in{\rm span}\{\ket{\beta_i}\bra{\alpha_j}:0\leq i\leq p-1,~0\leq j\leq q-1\}$, 
$B_{p\times (d-q)}\in{\rm span}\{\ket{\beta_i}\bra{\alpha_j}:0\leq i\leq p-1,~q\leq j\leq d-1\}$, 
$C_{(d-p)\times q}\in{\rm span}\{\ket{\beta_i}\bra{\alpha_j}:p\leq i\leq d-1,~0\leq j\leq q-1\}$ 
and $D_{(d-p)\times (d-q)}\in{\rm span}\{\ket{\beta_i}\bra{\alpha_j}:p\leq i\leq d-1,~q\leq j\leq d-1\}$. 
Notice that $D_{(d-p)\times (d-q)}={\bf 0}_{(d-p)\times (d-q)}$ for any $E'\in \cS'$, since any 
matrix $E\in \cS$ maps $U$ into a subspace of $\cS(U)$. Therefore, we conclude that any matrix 
in $\cS$ is transformed by $Q_1$ and $Q_2$ to the following form:
\begin{equation}\label{eq:compression form}
\left[
\begin{array}{c|c}
A_{p\times q} &  B_{p\times (d-q)} \\
\hline
 C_{(d-p)\times q}& {\bf 0}_{(d-p)\times (d-q)}
\end{array}
\right]_{d\times d}.
\end{equation}

Given $p$, $q$ satisfying $p+q<d$, all matrices of form~\ref{eq:compression 
form} span a matrix space, denoted by $\A(p,q,d)$. 
Clearly, $\A(p,q,d)$ has shrunk subspaces. In particular, any $d\times d$ matrix space $\cS$ with shrunk subspace $U$, satisfying $\dim(U)=d-q$ and $\dim(\cS(U))=p$, is a subspace of $\A(p,q,d)$ after the above transforming procedure. In this sense we may call $\cA(p, q, d)$ a maximal-compression matrix space. A formal definition is as follows.
\begin{definition}\label{def:maximal-compression matrix space}
Let $p,q\in\N$. The $d\times d$ matrix space $\A(p,q,d)$ is the matrix space 
spanned by those elementary matrices whose nonzero entry lies either in the first 
$p$ rows, or in the first $q$ columns, i.e.~
$$\A(p,q,d):={\rm span}\{\{\ket{i}\bra{j}:0\leq i\leq  p-1,~0\leq j\leq d-1\}\cup\{\ket{i}\bra{j}:p\leq i\leq d-1,~0\leq j\leq q-1\}\}.$$
The maximal rank of $\A(p,q,d)$ equals $\min\{p+q,d\}$. Moreover, if $p+q<d$, we say $\A(p,q,d)$ is a 
maximal-compression matrix space. 
\end{definition}
This definition can be generalized to the rectangular matrix spaces. Let $\A(p,q,m,n)$ be the $m\times n$ matrix space spanned by the first $p$ rows and $q$ columns of elementary matrices, i.e.~
$$\A(p,q,m,n):={\rm span}\{\{\ket{i}\bra{j}:0\leq i\leq  p-1,~0\leq j\leq n-1\}\cup\{\ket{i}\bra{j}:p\leq i\leq m-1,~0\leq j\leq q-1\}\}.$$
Then the maximal rank of $\A(p,q,m,n)$ equals $\min\{p+q,m,n\}$. And we say 
$\A(p,q,m,n)$ is a maximal-compression matrix space if 
$p+q<\min\{m,n\}$.

\subsection{Results on shrunk subspaces}

In this subsection we first review some mathematical results concerning shrunk 
subspaces. We then introduce recent progress on algorithms to decide the existence 
of shrunk subspaces. 

Shrunk subspaces emerge in several mathematical areas. The first appearance of 
shrunk subspaces seems to be
in T. G. Room's treatise on determinants in 1930's~\cite{Room}. We 
mentioned in 
Section~\ref{subsec:prel} that shrunk subspaces can be viewed as a linear 
algebraic analogue of shrunk subsets as in the Hall's marriage theorem, which in 
turn is a basic 
result in combinatorics~\cite{LP86}. In matroid theory, Lov\'asz observed that the 
intersection of two linear matroids naturally 
leads to shrunk subspaces~\cite{Lov89}.

Shrunk subspaces appear in non-commutative algebra as follows. Suppose $\cS\leq  M(d,\F)$ is a matrix space over some field $\F$ and spanned by $\{T_1, \dots, T_m\}\subset M(d,\F)$. Let $\{x_1, 
\dots, x_m\}$ be a set of non-commuting variables, and form a matrix 
$T=x_1T_1+\cdots +x_mT_m$ whose entries are linear forms in $x_i$'s. Matrices of 
this type have been studied in non-commutative algebra in the context of the free 
skew field since the 1970's~\cite{Cohn71}. The rank of such a matrix over the free skew 
field, which was named \emph{non-commutative rank} and denoted by $ncrk(\cdot)$, 
was shown to be the minimum $r$ such that there exists a subspace 
$U\leq \F^d$ with $\dim(U)-\dim(\cS(U))=d-r$~\cite{Fortin2004}. Thus, $ncrk(\cS)<d$ if and only if $\cS$ have shrunk subspaces.  

Another way to reach the concept of shrunk subspaces is to consider matrix spaces 
with maximal ranks bounded from above~\cite{Eisenbud1988}. Characterizing those 
matrix spaces is known to be a difficult problem; in fact, such matrix spaces 
basically correspond to certain torsion-free sheaves on 
projective spaces~\cite{Eisenbud1988}. To make progress on this topic, one 
approach is to consider certain ``witnesses'' that can serve as an upper bound on 
the maximal rank. As mentioned, shrunk subspaces can be used as such witnesses. 
Specifically, if a matrix space $\cS\leq  M(d,\F)$ has a shrunk subspace $U$ with 
$\dim(U)-\dim(\cS(U))=c>0$, then it is clear that $mrk(\cS)\leq d-c$. 

On the other hand, an interesting result by Fortin and Reutenauer showed the following: 
\begin{theorem}[Corollary 2. in Ref.~\cite{Fortin2004}]\label{ncrk}
Let $\cS$ be a matrix space in $ M(d, \F)$. Then $$mrk(\cS)\leq ncrk(\cS)\leq 
2mrk(\cS).$$
\end{theorem}

An important characterization of matrix spaces with shrunk subspaces comes from 
invariant theory. Consider the group action of $(A, C)\in \SL(d)\times \SL(d)$ on 
$ M(d,\F)^{\oplus m}$ by sending $(B_1, \dots, B_m)$ to $(AB_1C^T, \dots, AB_mC^T)$. 
This induces an action on the ring of polynomial functions on $ M(d,\F)^{\oplus 
m}$. Let $R(d, m)$ be the ring of those polynomials invariant under this 
action. $R(d, m)$ is called \emph{the ring of matrix semi-invariants} (for 
matrices of size $d\times d$)~\cite{Ivanyos2015a,DM2}. The common zeros of the homogeneous 
polynomials of positive degrees in 
$R(d, m)$, denoted as $N(R(d, m))$, is referred to as the
\emph{nullcone} of this invariant ring in the invariant theory literature. The 
 link to those matrix spaces which have shrunk subspaces is the following result from invariant theory, 
proved using the celebrated Hilbert-Mumford criterion. 
\begin{theorem}[\cite{BD06,ANK07}]\label{thm:nullcone}
$(B_1, \dots, B_m)$ is in $N(R(d, m))$ if and only if $\mathrm{span}\{B_1, \dots, 
B_m\}$ has a shrunk subspace. 
\end{theorem}
Therefore, matrix spaces with shrunk subspaces are 
characterized by those polynomials in $R(d, m)$. 

Invariant theory also helps to certify those matrix spaces with no shrunk 
subspaces. That is, given a matrix space $\cS\leq  M(d,\F)$, if $\cS$ does not have 
shrunk subspace, we would like to present a short witness to certify this fact. 
For example, if $\cS$ contains a full-rank matrix $A$, then exhibiting $A$ is 
enough to certify that $\cS$ does not contain shrunk subspaces. However, as 
mentioned, it is possible that a matrix space has neither full-rank matrices nor 
shrunk subspaces. To resolve this difficulty, we first recall what 
polynomials in $R(d, m)$ look like. This task is usually resolved in the so-called 
first fundamental theorem for $R(d, m)$. 
\begin{theorem}[\cite{DW00,SV01,DZ01,ANK07}]\label{thm:fft}
Every homogeneous polynomial in $R(d, m)$ is of degree $kd$ for some $k\in\N$, and 
is a linear 
combination of polynomials of the form $\det(X_1\otimes A_1+\cdots+X_m\otimes A_m)$ 
where the $X_i$'s are $d\times d$ variable matrices, and the $A_i$'s are $k\times k$ 
matrices over $\F$. 
\end{theorem}
Theorem~\ref{thm:fft} motivates the following definition. For a matrix space 
$\cS\leq M(d,\F)$, the $k$th \emph{blow-up} of $\cS$ is $\cS^{[k]}= \cS\otimes M(k,\F)$. It is clear that, if $\cS$ possesses a shrunk subspace, then 
$\cS^{[k]}$ has a shrunk subspace for any positive integer $k$. On the other hand,
if $\cS=\mathrm{span}\{B_1, \dots, B_m\}$ does not possess a shrunk subspace, then it is 
not in the nullcone of $R(d, m)$ (Theorem~\ref{thm:nullcone}). This implies that 
there exists some $A_1, \dots, A_m\in  M(k,\F)$ such that $\det(B_1\otimes 
A_1+\dots+B_m\otimes A_m)\neq 0$ (Theorem~\ref{thm:fft}), which just says that 
$\cS^{[k]}$ contains a non-singular matrix. To see that $k$ is finite is 
classical: by Hilbert's basis theorem,  
$N(R(d, m))$ can be defined by finitely many polynomials, therefore $k$ is also 
finite. Recently, exciting progress suggests that $k$ can be taken to be no more 
than $d-1$ as long as $|\F|$ is large enough~\cite{Derksen201744}; see also~\cite{IQS2} for 
a simpler proof of $k\leq d+1$. Summarizing the above we have 
\begin{theorem}[\cite{Derksen201744}]\label{thm:deg_bd}
Suppose $|\F|=d^{\Omega(1)}$, where $\Omega(1)$ is asymptotic in $d$. If $\cS\leq  M(d,\F)$ does not have shrunk subspace, then for some $k\leq d-1$, $\cS\otimes  M(k,\F)$ contains a full-rank matrix. 
\end{theorem}
The full-rank matrix as in Theorem~\ref{thm:deg_bd} then serves as a short witness 
for $\cS$ to have no shrunk subspace. We can also easily formulate an algorithmic problem around shrunk subspaces as 
follows. 
\begin{problem}\label{prob:non-commutative rank problem}
Given a matrix space $\cS\leq M(d,\F)$, decide whether 
$\cS$ has a shrunk subspace $U\leq \F^d$.
\end{problem}
This problem is known as the non-commutative rank problem~\cite{Fortin2004}, the 
non-commutative rational identity testing problem~\cite{Hrubes2015}, and the 
non-commutative Edmonds' problem~\cite{IQS1}. We adopt the non-commutative 
rank 
problem as in~\cite{Fortin2004}. This name choice is due to the connection with 
the free skew field as mentioned above. Recent advances imply that this problem 
can be solved deterministically in polynomial time.
\begin{theorem}[Ref.~\cite{Garg2015,Ivanyos2015a}]\label{polynomialtime}
There exists a deterministic polynomial-time algorithm to decide whether a given 
matrix space $\cS\leq  M(d,\F)$ has full non-commutative rank or not when $|\F|$ is large enough.
\end{theorem}

\section{Multi-copy transformation}\label{finite}
In this section, we discuss the super-multiplicativity of the maximal Schmidt rank. We say a map $f:\Hm_d\to\R$ is super-multiplicative, if for any $\ket{\psi}$, $\ket{\phi}\in\Hm_d$, $f(\ket{\psi}\ox\ket{\phi})\geq f(\ket{\psi})\times f(\ket{\phi})$. Clearly, the maximal Schmidt rank, as well as many other information theoretic quantities, are super-multiplicative. A special case of the super-multiplicative maps is the multiplicative maps, i.e.~those maps $f:\Hm_d\to\R$ satisfy $f(\ket{\psi}\ox\ket{\phi})=f(\ket{\psi})\times f(\ket{\phi})$ for any $\ket{\psi},\ket{\phi}\in\Hm_d$.
Multiplicative maps are important in quantum in formation theory, as they provides 
computable asymptotic quantities. For instance, the 
\emph{negativity}~\cite{vidal2002computable} is an 
important computable entanglement measure, which provides an upper bound on 
distillable entanglement. Similarly, several multiplicative quantities have been constructed 
to provide efficiently computable bounds on distillable 
entanglement~\cite{wang2016improved}, entanglement cost~\cite{wang2016rains}, 
and classical and quantum capacity of quantum 
channels~\cite{holevo2001evaluating,Duan2013,wang2016semidefiniteQ,wang2016semidefiniteC}.

On the other hand, many information-theoretic quantities are not multiplicative, 
but strictly super-multiplicative. A \emph{strictly super-multiplicative map} $f:\Hm_d\to\R$ is super-multiplicative, and there exist $\ket{\psi}$, $\ket{\phi}\in\Hm_d$ satisfying $f(\ket{\psi}\ox\ket{\phi})>f(\ket{\psi})\times f(\ket{\phi})$. This property causes much 
difficulty to compute the corresponding asymptotic counterparts, and yields many intriguing phenomenons. One example is that, two quantum channels with vanishing quantum capacities can have nonzero quantum capacity when used together~\cite{smith2008quantum}. This phenomenon not only illustrates the super-additivity of quantum capacity (which becomes super-multiplicativity before taking logarithm), but also illustrated that quantum capacity can be ``super-activated''. In the rest, we show that the maximal Schmidt rank is strictly super-multiplicative. 
\begin{theorem}\label{skew sym}
Let $d$ be an odd number and  
$$\ket {{\rm\Psi}^d_{\rm ABC}}:=\sqrt{\frac{2}{d(d-1)}} 
\sum_{0\leq i<j\leq 
d-1}(\ket{i}\ket{j}-\ket{j}\ket{i})\ox\ket{\psi_{ij}}\in\Hm^{\rm A}_d\ox\Hm^{\rm B}_d\ox\Hm^{\rm C}_{d^2},$$ 
where $\{\ket{i}:0\leq i\leq d-1\}$, $\{\ket{j}:0\leq j\leq d-1\}$ and $\{\ket{\psi_{ij}}:0\leq i,j\leq d-1\}$ are sets of orthogonal basis of $\Hm^{\rm A}_d$, $\Hm^{\rm B}_d$ and $\Hm^{\rm C}_{d^2}$, respectively. We have
\begin{equation}\label{eq:super-multiplicativity}
msrk(\ket {{\rm\Psi}^d_{\rm ABC}})=d-1,~msrk(\ket {{\rm\Psi}^d_{\rm ABC}}^{\ox 2})=d^2>(d-1)^2.
\end{equation}
In fact, any tripartite state $\ket {{\rm\Phi}^d_{\rm ABC}}$ with $M({\rm\Phi}^d_{\rm ABC})={\rm span}\{\ket{i}\bra{j}-\ket{j}\bra{i}: 0\leq i<j\leq d-1\}$, which is the $d\times d$ skew-symmetric matrix space, satisfies equation~(\ref{eq:super-multiplicativity}).  
\end{theorem}

\begin{proof}
It is known that for odd $d$, the maximal rank of the $d\times d$ skew-symmetric matrix space is $d-1$~\cite{Fortin2004}.
We now show 
$mrk(M({\rm\Phi}^d_{\rm ABC}))=d^2$.
Specifically, let $E_{i,j}=\ket{i}\bra{j}-\ket{j}\bra{i}\in M({\rm\Phi}^d_{\rm 
ABC})$ for $0\leq i<j\leq d-1$. We claim that 
$$P:=\sum_{0\leq i<j\leq d-1}E_{i,j}\otimes E_{i,j}$$ has rank $d^2$. Notice that $P$ 
is in the block matrix form:
$$P=\begin{pmatrix}
{\bf 0}&E_{0,1}&\cdots&E_{0,d-1} \\
-E_{0,1}&{\bf 0}&\cdots& E_{1,d-1} \\
\vdots&\vdots&\ddots&\vdots\\
-E_{0,d-1}&-E_{1,d-1}&\cdots&{\bf 0}\\
\end{pmatrix}.$$
We will prove that $\kernel(P)=\{0\}$. Let $\ket{\alpha}=\sum_{i,j=0}^{d-1} 
x_i(j)\ket{i}\ket{j}$ such that $P\ket{\alpha}=0$, where $x_i(j)$ are variables. Denote $\ket{\alpha_i}=\sum_{j=0}^{d-1}x_i(j)\ket{j}$. Then for 
$1\leq k\leq d-2$, we have:
$$-\sum_{i=0}^{k-1}E_{i,k}\ket{\alpha_i}+\sum_{i=k+1}^{d-1}E_{k,i}\ket{\alpha_i}=0.$$
For $k=0$, we have $\sum_{i=1}^{d-1}E_{0,i}\ket{\alpha_i}=0$ and for $k=d-1$, we have $\sum_{i=0}^{d-2}E_{i,d-1}\ket{\alpha_i}=0$. Noticing that $E_{j,k}\ket{\alpha_i}=x_{i}(k)\ket{j}-x_i(j)\ket{k}$, we can rewrite the equations to:
\begin{align*}
\sum_{i=1}^{d-1}(x_i(0)\ket{i}-x_i(i)\ket{0})&=0,\\
-\sum_{i=0}^{k-1}(x_i(k)\ket{i}-x_i(i)\ket{k})+\sum_{i=k+1}^{d-1}(x_i(i)\ket{k}-x_i(k)\ket{i})&=0,~~k=1,\dots,d-2\\
\sum_{i=0}^{d-2}(x_i(d-1)\ket{i}-x_i(i)\ket{d-1})&=0.
\end{align*}
These yield that $x_i(j)=0$ for all $0\leq i\neq j\leq d-1$ and $m_k=\sum_{i\neq k}x_i(i)=0$ for $k=0,\dots,d-1$. Notice that $m_{k}-m_{k+1}=x_{k+1}(k+1)-x_{k}(k)=0$ for $k=0,\dots, d-2$ and $m_{d-1}-m_0=x_{0}(0)-x_{d-1}(d-1)=0$. We derive $x_k(k)=0$ for $k=0,\dots,d-1$.
Thus $\ket{\alpha}=0$ is the only solution for $P\ket{\alpha}=0$, hence ${\rm rank}(P)=d^2$. 
\end{proof}

Theorem~\ref{skew sym} also implies that $\ket {{\rm\Psi}^d_{\rm ABC}}$ cannot be 
transformed to the $d\otimes d$ maximally entangled state by means of SLOCC, but 
can do so with two copies. In the multi-copy regime, those 
tripartite states $\ket{{\rm\Psi}_{\rm ABC}}$ satisfying $msrk(\ket{{\rm\Psi}_{\rm 
ABC}}^{\ox 2})>msrk(\ket{{\rm\Psi}_{\rm ABC}})^2$ are of great interest, as such 
states allow for more advantages in the tripartite-to-bipartite SLOCC 
transformations when multiple copies are provided.
Interestingly, these tripartite states can be characterized by the following theorem, based on the structure of their corresponding matrix spaces:

\begin{theorem}\label{super}
Given a tripartite state $\ket{{\rm \Psi}_{\rm ABC}}\in \Hm^{\rm A}_{d_1}\ox\Hm^{\rm B}_{d_2}\ox\Hm^{\rm C}_{d_3}$, let $\cS=M({\rm \Psi}_{\rm ABC})\leq M(d_1\times d_2,\C)$. Then $msrk(\ket{{\rm \Psi}_{\rm ABC}}^{\ox 2})>msrk(\ket{{\rm \Psi}_{\rm ABC}})^2$ if and only if 
\begin{enumerate}
\item $mrk(\cS)< \dims[\im(\cS)]$, and
\item $mrk(\cS)< d_2 -\dims[{\rm Ker}(\cS)]$,
\end{enumerate}
where the image and kernel of a matrix space 
$\cS$ are defined as $\im(\cS):={\rm span}\{\cup_{\textit{E}\in \cS} \im(\textit{E})\}$ and 
${\rm Ker}(\cS):=\cap_{\textit{E}\in \cS}{\rm Ker}(\textit{E})$.
\end{theorem}
\begin{remark}
Theorem~\ref{super} indicates that the inequality $msrk(\ket{{\rm \Psi}_{\rm ABC}}^{\ox 2})>msrk(\ket{{\rm \Psi}_{\rm 
ABC}})^2$ holds for almost all tripartite states whose maximal Schmidt rank is not full, as the two conditions in theorem~\ref{super} put only degenerate restrictions on such states. 
\end{remark}

\begin{proof}
Equivalently, we consider when $mrk(\cS^{\ox 2})>mrk(\cS)^2$. Without loss of generality, we assume $d_1\leq d_2$. The following observation from Ref.~\cite{Ivanyos2010} is useful to estimate the rank of linear combinations of two 
matrices: 
\begin{lemma}[Lemma 2.2 in Ref.~\cite{Ivanyos2010}]\label{rank}
Given two matrices ${\textit X, Y}\in M(d_1\times d_2,\C)$. If $Y\kernel(\textit{X})\not\leq \im(\textit{X})$, then $\rm{rank}(\textit{X+rY})>\rm{rank}(\textit{X})$ except for at most $\rm{rank}(\textit{X})+1$ elements $r\in\C$.
\end{lemma}
We first prove the sufficiency. 
Notice that a matrix space 
$\cS$ satisfying the above two conditions must be singular. Choose $E\in\cS$ 
with the highest rank, i.e.~$\rank(E)=mrk(\cS)<d_1$. Define the following two matrix 
spaces:
$$\cX:=\{F\in\cS : \im(\textit{F})\leq \im(\textit{E})\},~~\cY:=\{F\in\cS : \kernel(\textit{E})\leq\kernel(\textit{F})\}.$$
We claim that $\cX$ and $\cY$ are two proper subspaces of $\cS$. Otherwise, assuming $\cX=\cS$, we have $mrk(\cS)=\rank(\textit{E})=\dims[\im(\textit{E})]=\dims[\im(\cS)]$, which is a contradiction. If $\cY=\cS$, we have $\dims[\kernel(\cS)]=\dims[\kernel(\textit{E})]=d_2-\rank(\textit{E})=d_2-mrk(\cS)$, which is also a contradiction. Now we can choose $E'\in \cS$ such that $E'\notin \cX\cup \cY$, i.e.~$\im(\textit{E}')\not\leq\im(\textit{E})$ and $\kernel(\textit{E})\not\leq\kernel(\textit{E}')$. 

Then we claim that there exists $r\in\C$ such that $\rank({ E\otimes E}+r{ E'\otimes E'})>\rank({ E\otimes E})$. Notice that $\kernel(E\otimes E)={\rm span}\{\kernel(E)\otimes \C^{d_2},\C^{d_2}\otimes \kernel(E)\}$ and $\im(E\otimes E)=\im(E)\otimes\im(E)$. 
We are going to show $(E'\otimes E')(\kernel(E)\otimes \C^{d_2})\not\leq \im(E)\otimes\im(E)$.
Since $E$ has the highest rank in $\cS$, 
$\rank(\textit{E+rE}')\leq\rank(\textit{E})$ for any $r\in\C$. By Lemma~\ref{rank}, we see that $E'\kernel(\textit{E})\leq\im(\textit{E})$. Since $\kernel(\textit{E})\not\leq \kernel(\textit{E}')$, we 
can choose a non-zero vector $\ket{v}\in\kernel(\textit{E})$ such that $0\neq 
E'\ket{v}\in\im(\textit{E})$. Moreover, since 
$\im(\textit{E}')\not\leq\im(\textit{E})$, we can find a vector $\ket{u}$ such 
that $0\neq E'\ket{u}\not\in\im(\textit{E})$. Setting $\ket{v}\otimes 
\ket{u}\in\kernel(\textit{E})\otimes\C^{d_2}\leq\kernel(E\otimes E)$, we have $0\neq E'\ket{v}\otimes 
E'\ket{u}\not\in\im(\textit{E})\otimes\im(\textit{E})$. Then by Lemma~\ref{rank}, there exists $r\in\C$ 
such that:
$$mrk(\cS^{\otimes 2})\geq\rank(E\otimes E+rE'\otimes E')>\rank(E\otimes E)=mrk(\cS)^2.$$

For the necessity, we shall show that if $\cS$ does not satisfy either one of those two 
conditions, $mrk(\cS^{\otimes 2})=mrk(\cS)^2$. When $\cS$ is non-singular, this 
equation always holds. Now give a singular matrix space $\cS\leq M(d_1\times d_2,\C)$ satisfying $\rk(\cS)=\dims[\im(\cS)]$, and $mrk(\cS)=d'<d_1$. 
Fix bases of $\im(\cS)$, say $\{v_1,\dots,v_{d'}\}$ and 
extend them to $\{v_1,\dots,v_{d'},v_{d'+1},\dots,v_{d_1}\}$, which are full bases of $\C^{d_1}$. 
For any matrix $F\in\cS$, we have $\im(\textit{F})\leq{\rm span}\{v_1,\dots,v_{d'}\}$. 
Then any matrix $F\in\cS$ can be expressed by this basis in the following form:
$$F=\begin{pmatrix} F_{d'\times d_2} \\ {\bf 0}_{(d_1-d')\times d_2}\end{pmatrix},$$
where $F_{d'\times d_2}$ is a ${d'}\times d_2$ matrix and the rest part of $F$ are all zero. Thus we 
can derive $\rk(\cS^{\otimes 2})\leq\rk(\cS)^2$ since any matrix in $\cS\otimes\cS$ 
will have at most ${d'}^2$ non-zero rows. The other case is similar, and we conclude the proof.  
\end{proof}

\section{Asymptotic transformations}\label{proof}
Theorem~\ref{skew sym} implies the super-multiplicativity of 
the maximal Schmidt rank. And theorem~\ref{super} indicates how the structure of matrix spaces 
plays a role in this topic.  In this section, we initiate the study of the 
asymptotic maximal (Schmidt) rank. Utilizing certain results from the structure of matrix spaces, 
including some relevant results from invariant theory, we are able to present explicit 
formulas to calculate this asymptotic quantity for two important families of matrix spaces (tripartite states), 
which would be difficult to compute without the help of these tools. Our results are summarized in the following theorem:
\begin{theorem}\label{asymptotic}
Given a matrix space $\cS\leq M(d,\C)$,
\begin{enumerate}
\item If $\cS$ does not have shrunk subspace, we have $\frac{1}{2}d^n\leq mrk(\cS^{\ox n})\leq d^n$. Thus $mrk^{\infty}(\cS)=d$.

\item If $\cS=\A(p,q,d)$ satisfying $p+q<d$, then 
\begin{equation}\label{eq:asymptotic formula}
mrk^{\infty}(\cS)= d\max\{2^{-D(1-\alpha || p')}, 2^{-D(\alpha || q')}\},
\end{equation}
where $p'=\frac{p}{d}$, $q'=\frac{q}{d}$, $\alpha=\frac{\log_2(d-q)-\log_2 p}{\log_2((d-p)(d-q))-\log_2 (pq)}$ and $D(a||b):=a{\rm log}_2\frac{a}{b}+(1-a){\rm log}_2\frac{1-a}{1-b}$. 
\end{enumerate}
\end{theorem}

\noindent\textit{Proof of theorem~\ref{asymptotic}. 1.} 
Recall that if $\cS$ does not have shrunk subspace, $ncrk(\cS)=d$. Then by theorem~\ref{ncrk}, $\frac{1}{2}d\leq mrk(\cS)\leq d$. Thus, we first prove that, given two matrix spaces without shrunk subspaces, there tensor product still has no shrunk subspaces. 
\begin{lemma}\label{exceptional}
Given two matrix spaces $\cS_1\leq M(d_1,\C)$ and $\cS_2\leq M(d_2,\C)$ which have no shrunk subspace. Then $\cS_1\otimes\cS_2\leq  M(d_1d_2,\C)$ has no shrunk subspace.
\end{lemma}
\begin{proof}
Recall that by theorem~\ref{thm:deg_bd}, for $i=1,2$, if $\cS_i\leq M(d_i,\C)$ has no shrunk 
subspace, there exists $k_i\leq d_i-1$, such that $mrk(\cS\otimes M(k_i,\C))=k_id_i$. Now consider the matrix space 
$\cS_1\otimes\cS_2\otimes M(k_1k_2,\C)$, which is isomorphic to 
$\cS_1\otimes M(k_1,\C)\otimes\cS_2\otimes M(k_2,\C)$. Since $\cS_1$ and $\cS_2$ have no 
shrunk subspace, we can find matrices $A_i\in\cS_i\otimes M(k_i,\C)$ such that ${\rm 
rank}(A_i)=k_id_i$ for $i=1,2$. Then we have 
$mrk(\cS_1\otimes\cS_2\otimes M(d_1d_2,\C))={\rm rank}(A_1\ox A_2)=k_1k_2d_1d_2$. By theorem~\ref{thm:nullcone} and theorem~\ref{thm:fft}, $\cS_1\ox \cS_2$ cannot have shrunk subspaces. 
\end{proof}

By this lemma, we derive that for $\cS\leq M(d,\C)$ satisfying $ncrk(\cS)=d$, 
$ncrk(\cS^{\otimes n})=d^n$ for $n\in\N$. By theorem~\ref{ncrk},
we have $\frac{1}{2}d^n\leq mrk(\cS^{\otimes n})\leq d^n$, and $mrk^\infty(\cS)=d$ 
follows.\qed
\medskip

To prove theorem~\ref{asymptotic}. 2, we first explicitly compute the maximal rank of the tensor product of two maximal-compression matrix spaces: 
\begin{lemma}\label{compression}
Given two maximal-compression matrix spaces $\A_1=\A(p_1,q_1,m_1,n_1)$ and $\A_2=\A(p_2,q_2,m_2,n_2)$ (satisfying $p_1+q_1<\min\{m_1,n_1\}$ and $p_2+q_2<\min\{m_2,n_2\}$), we have:
\begin{equation}\label{eq:rank of tensor product}
mrk(\A_1\otimes \A_2)=p_1p_2+\min\{(n_1-q_1)q_2,p_1(m_2-p_2)\}+\min\{(m_1-p_1)p_2,q_1(n_2-q_2)\}+q_1q_2.
\end{equation}
\end{lemma}
\begin{proof}
It is convenient to view $\cA(p,q,m,n)$ as symbolic matrix $P$, of which the entries are filled with $0$ and $*$. More precisely, the $(i,j)$th entry of the symbolic matrix is $*$ if and only if there exist matrices in $\cA(p,q,m,n)$ such that the $(i,j)$th entry of which is non-zero. Due to the structure of $\cA(p,q,m,n)$, these $*$s can be assigned arbitrary complex numbers independently. Moreover, it is easy to see that, the symbolic matrix of $\cA(p_1,q_1,m_1,n_1)\otimes\cdots\otimes\cA(p_k,q_k,m_k,n_k)$ is $P_1\otimes \cdots\otimes P_k$, where $P_i$ is the symbolic matrix of $\cA(p_i,q_i,m_i,n_i)$ for $i=1,\dots, k$ and the multiplication rule of $\{0,*\}$ is $0\times 0=0$, $0\times *=*\times 0=0$, $*\times *=*$.

With this definition, we firstly write down the symbolic matrix $P$ of $\A_1\otimes\A_2$:
\begin{equation}\label{eq:matrix form}
P=
\begin{pmatrix}
A_{1,1} & \cdots & A_{1,q_1} & A_{1,q_1+1} & \cdots  & A_{1,n_1} \\
\vdots & P_0 & \vdots & \vdots & P_1  &  \vdots \\
A_{p_1,1} & \cdots & A_{p_1,q_1} & A_{p_1,q_1+1} & \cdots  & A_{p_1,n_1} \\
A_{p_1+1,1} & \cdots & A_{p_1+1,q_1} &  & & & \\

\vdots & P_2 & \vdots &  & \textbf{\huge 0}& \\
A_{m_1,1} & \cdots & A_{m_1,q_1} &  & & \\
\end{pmatrix},
\end{equation}
where $A_{i,j}$ is the symbolic matrix of $\cA_2$ for all possible $i$ and $j$, and the rest block of size $_{(m_1-p_1)m_2\times (n_1-q_1)n_2}$ are all zero. Denote the upper left block by $P_0$, the upper right block by $P_1$ and the lower left block by $P_2$. We will show 
that, after properly rearranging rows and columns, $P_0$, $P_1$ and $P_2$ become symbolic matrices of $\cA(p,q,m,n)$ with different parameters.  This can be done by the follows: In $P_1$, we move all columns with more than $p_1p_2$ $*$s to the left and move all rows with more than $(n_1-q_1)q_2$ $*$s to the top. In $P_2$, we move all rows with more than $q_1q_2$ $*$s to the top and move all columns with more than $(m_1-p_1)p_2$ $*$s to the top. These row and column rearrangements are equivalent to left and right multiplying with invertible matrices $Q_1\in M(m_1m_2,\C)$ and $Q_2\in M(n_1n_2,\C)$, respectively. More precisely, let $P'=\begin{pmatrix}P_0'&P_1'\\P_2'&{\bf 0}\end{pmatrix}$ be the symbolic matrix of $\cA'=Q_1(\cA_1\ox\A_2)Q_2$. Then $P_1'$ is the symbolic matrix of $\cA_1'=\A(p_1p_2,(n_1-q_1)q_2,p_1m_2,(n_1-q_1)n_2)$ and $P_2'$ is the symbolic matrix of $\cA_2'=\A((m_1-p_1)p_2,q_1q_2,(m_1-p_1)m_2,q_1n_2)$. Moreover, it is easy to verify that $P_0'$ is the symbolic matrix of $\cA_0'=\cA(p_1p_2,q_1q_2,p_1m_2,q_1n_2)$.

Then we prove that $mrk(\A_1\otimes\A_2)=mrk(\A_1')+mrk(\A_2')$.  
Firstly, we show that there exists 
$P'=\begin{pmatrix}P_0'&P_1'\\P_2'&{\bf 0}\end{pmatrix}\in\cA'$ satisfying 
$rank(P')=mrk(\cA')=mrk(\cA_1\ox\cA_2)$, 
$rank(P_1')=mrk(\cA_1')$ and $rank(P_2')=mrk(\cA_2')$. Notice that there always exists
$P''\in \cA'$ with $rank(P'')=mrk(\cA')$, 
$R'=\begin{pmatrix}R_0'&R_1'\\R_2'&{\bf 0}\end{pmatrix}\in \cA'$ with 
$rank(R_1')=mrk(\cA_1')$ and 
$R''=\begin{pmatrix}R_0''&R_1''\\R_2''&{\bf 0}\end{pmatrix}\in \cA'$ with 
$rank(R_2'')=mrk(\cA_2')$. We claim that there exist $\alpha, \beta, \gamma\in 
\C$, such that $P'=\alpha P''+\beta R'+\gamma R''$ is what we need. To see this, 
consider the matrix $xP''+yR'+zR''$, where $x,y,z$ are variables. As $rank(P'')=mrk(\cA')=r$, there exists an $r\times r$ submatrix of $P''$ with rank $r$. 
Let $f_1$ be the determinant of the corresponding submatrix in $xP''+yR'+zR''$. $f_1$ is a 
nonzero homogeneous polynomial in $\C[x, y, z]$ of degree $r$. Similarly, let 
$s=mrk(\cA_1')$ and $t=mrk(\cA_2')$. Then there exists an $s\times s$ (resp. 
$t\times t$) submatrix of $xP''+yR'+zR''$ in the upper right (resp. lower left) part, 
such that, if we denote its determinant by $f_2$ (resp. $f_3$), then $f_2$ (resp. 
$f_3$) is a nonzero 
homogeneous polynomial in $\C[x, y, z]$ of degree $s$ (resp. $t$). Since
$f=f_1f_2f_3$ is a nonzero polynomial in $\C[x,y,z]$, there exists $(\alpha, 
\beta, \gamma)\in \C^3$ such that $f(\alpha, \beta, \gamma)\neq 0$. Such $(\alpha, 
\beta, \gamma)$ then translates to our desired conditions for $\alpha P''+\beta 
R'+\gamma R''$.

Take such $P'=\begin{pmatrix}P_0'&P_1'\\P_2'&{\bf 0}\end{pmatrix}\in\cA'$.
Since $p_1+q_1<\min\{m_1,n_1\}$ and $p_2+q_2<\min\{m_2,n_2\}$, we have $p_1p_2<(n_1-q_1)(n_2-q_2)$ and $q_1q_2<(m_1-p_1)(m_2-p_2)$. 
Then submatrix in the upper right part of $P_1'$ has full row rank $p_1p_2$,
and the lower left part of $P_2'$ has full column rank $q_1q_2$. For any 
$P_0'\in\cA(p_1p_2,q_1q_2,p_1m_2,q_1n_2)$, we can use the upper right part of 
$P_1'$ to clear the first $p_1p_2$ rows of $P_0'$ without 
changing the rank of $P'$. Similarly, we can use the lower left part of $P_2'$ to clear the first 
$q_1q_2$ columns of $P_0'$ without changing the rank of $P'$. After these row and column 
operations, $P'$ is transformed to $\begin{pmatrix}{\bf 0}&P_1'\\P_2'&{\bf 0}\end{pmatrix}$. 
This then shows that
\begin{equation}
\begin{split}
&mrk(\A_1\otimes \A_2)=rank(P')=rank(P_1')+rank(P_2')\\
=&p_1p_2+\min\{(n_1-q_1)q_2,p_1(m_2-p_2)\}+\min\{(m_1-p_1)p_2,q_1(n_2-q_2)\}+q_1q_2.
\end{split}
\end{equation}
\end{proof}

Let us examine an example to illustrate the above procedure. Consider 
$\cA_1=\cA_2=\cA(1,1,3,3)=\spa\{\proj{0},\ket{0}\bra{1},\ket{0}\bra{2},\ket{1}\bra{0},\ket{2}\bra{0}\}$,
 where $\{\ket{0},\ket{1},\ket{2}\}$ is the computational basis of $\Hm_3$, and 
the linear span is taken over $\C$. It is easy to see that 
$mrk(\cA(1,1,3,3))=2$. We show how to use lemma~\ref{compression} to compute the maximal 
rank of $\cA(1,1,3,3)^{\ox 2}$. Firstly, we exchange rows and columns to obtain an 
equivalent matrix space $\cA'$ of $\cA(1,1,3,3)^{\ox 2}$. This can be done by choosing 
$Q=\proj{00}+\proj{01}+\proj{02}+\proj{10}+\ket{11}\bra{20}+\proj{12}+\ket{20}\bra{11}+\proj{21}+\proj{22}.$
Then $\cA'=Q\cA(1,1,3,3)^{\ox 2}Q$ follows. More specifically, let $P$ be 
the symbolic matrix in $\cA(1,1,3,3)^{\ox 2}$. Multiplying $Q$ with $P$ form the left 
exchanges the $5$th row with the $7$th row of $P$; 
multiplying $Q$ with $P$ from the right exchanges the $5$th column with the $7$th column. 
So letting $P'=QPQ$, we have
$$P=\begin{pmatrix}
*&*&*&*&*&*&*&*&*\\
*&0&0&*&0&0&*&0&0\\
*&0&0&*&0&0&*&0&0\\
*&*&*&0&0&0&0&0&0\\
*&0&0&0&0&0&0&0&0\\
*&0&0&0&0&0&0&0&0\\
*&*&*&0&0&0&0&0&0\\
*&0&0&0&0&0&0&0&0\\
*&0&0&0&0&0&0&0&0\\
\end{pmatrix}~\Rightarrow
P'=\begin{pmatrix}
*&*&*&*&*&*&*&*&*\\
*&0&0&*&*&0&0&0&0\\
*&0&0&*&*&0&0&0&0\\
*&*&*&0&0&0&0&0&0\\
*&*&*&0&0&0&0&0&0\\
*&0&0&0&0&0&0&0&0\\
*&0&0&0&0&0&0&0&0\\
*&0&0&0&0&0&0&0&0\\
*&0&0&0&0&0&0&0&0\\
\end{pmatrix}.$$
Denote $P_0$ to be its submatrix of size $3\times 3$ in the upper left corner, $P_1$ to be its submatrix of size $3\times 6$ in the upper right corner, $P_2$ to be its submatrix of size $6\times 3$ in the lower left corner. It is easy to see that $P_0$ is the symbolic matrix of  
$\cA(1,1,3,3)$, $P_1$ is the symbolic matrix of $\cA(1,2,3,6)$, and $P_2$ is the symbolic matrix of $\cA(2,1,6,3)$. Applying Lemma~\ref{compression}, $mrk(\cA(1,1,3,3)^{\otimes 
2})=1+2+2+1=6$, which can be verified easily by looking at the form of $P'$.

This example also shows that the matrix spaces formed by the anti-diagonal blocks of $P'$, namely 
$\cA(1,2,3,6)$ and $\cA(2,1,6,3)$, may not form maximal-compression matrix spaces, as 
$1+2=\min\{3, 6\}$. (Recall that for $\cA(p, q, m, n)$ to be a maximal-compression matrix space 
we need $p+q<\min\{m, n\}$.) Thus lemma~\ref{compression} cannot be directly 
applied 
to capture a general formula when taking tensor product multiple times. 
Fortunately, for the case that $m=n=d$, we 
can evaluate the maximal rank of the $N$th tensor powers of a 
maximal-compression matrix space by the following: 
\begin{lemma}\label{lemma:rank formula}
Given a maximal-compression matrix space $\A(p,q,d)$ and an integer $N\geq 0$, the maximal rank of $\A(p,q,d)^{\otimes N+1}$ equals
\begin{equation}\label{rankformula}
\sum_{k=0}^N \binom{N}{k} \Big(\min\{p^{N-k+1}(d-p)^{k}, 
q^{k}(d-q)^{N-k+1}\}+\min\{q^{k+1}(d-q)^{N-k},p^{N-k}(d-p)^{k+1}\}\Big).
\end{equation}
\end{lemma}
\begin{proof}
The proof idea is as follows: First, we use induction to show that the symbolic matrix of $\A(p,q,d)^{\ox N}$ is in upper-anti-block-diagonal 
form, after appropriate row and column rearrangements. Notice that all these block matrices either equals zero matrix, or form $\cA(p,q,m,n)$s with different parameters. Second, we explicitly compute the maximal rank of those anti-diagonal $\cA(p,q,m,n)$s. Combining these two observations, we use the similar techniques used in lemma~\ref{compression} to prove this lemma.

The structure of matrices in $\A(p,q,d)^{\ox N}$ can be shown by the following:
\begin{observation}\label{claim:decompose matrix space}
For $N\geq 1$, there exist invertible matrices $Q_1\in M(d^{N},\C)$ and $Q_2\in M(d^N,\C)$, such 
that the symbolic matrix $P$ of $\cA'=Q_1\cA(p,q,d)^{\ox N}Q_2$ is of upper-anti-block-diagonal form: 
\begin{equation}\label{eq:anti-diagonal matrix form}
P=\begin{pmatrix}
P_{0,2^{N-1}-1} & \cdots & P_{0,l} & \cdots &P_0\\
\vdots & \reflectbox{$\ddots$} & \reflectbox{$\ddots$} &\reflectbox{$\ddots$}& {\bf 0}   \\
P_{l,2^{N-1}-1} & \reflectbox{$\ddots$} &P_l& \reflectbox{$\ddots$} & {\bf 0}\\
\vdots &\reflectbox{$\ddots$}& \reflectbox{$\ddots$} & \reflectbox{$\ddots$} & 
\vdots\\
P_{2^{N-1}-1}& {\bf 0} & {\bf 0} & \cdots  & {\bf 0}
\end{pmatrix}.
\end{equation}
\begin{enumerate}
\item For the anti-diagonal block matrices, label them as $P_0,\dots,P_{2^{N-1}-1}$. For an integer $l=0,\dots,2^{N-1}-1$, let $h(l)$ be the hamming weight of $l$, i.e.~the number of $1$'s in the binary expansion of $l$. Then $P_l$ is the symbolic matrix of 
$$\cA_l=\cA(p^{N-h(l)}(d-p)^{h(l)},q^{h(l)+1}(d-q)^{N-h(l)-1},p^{N-h(l)-1}(d-p)^{h(l)}d,q^{h(l)}(d-q)^{N-h(l)-1}d).$$

\item For the upper-left block matrices, label them as $P_{u,v}$ for $u,v\in\{0,2^{N-1}-1\}$, where $u$ is the label of anti-diagonal block matrix $P_{u}$ on the right of $P_{u,v}$ and $v$ is the label of anti-diagonal block matrix $P_{v}$ below $P_{u,v}$. If $h(u) \geq  h(v)$, $P_{u,v}={\bf 0}$. Otherwise $P_{u,v}$ is the symbolic matrix of 
$$\cA_{u,v}=\cA(p^{N-h(u)}(d-p)^{h(u)},q^{h(v)+1}(d-q)^{N-h(v)-1},p^{N-h(u)-1}(d-p)^{h(u)}d,q^{h(v)}(d-q)^{N-h(v)-1}d).$$
\end{enumerate}
\end{observation}
\begin{proof}[Proof of Observation~\ref{claim:decompose matrix space}]
We show the observation holds by induction on $N$. 
It holds for $N=1$ trivially. Assume for 
$\A(p,q,d)^{\ox N}$, observation~\ref{claim:decompose matrix space} holds. We consider 
$\A(p,q,d)^{\ox N+1}=\A(p,q,d)^{\ox N}\ox \cA(p,q,d)$. Firstly, let 
$\cA_1=(Q_1\otimes I_d)\A(p,q,d)^{\ox N+1}(Q_2\otimes I_d)=\cA'\otimes 
\A(p,q,d)$, 
where $Q_1$ and $Q_2$ are the matrices from the inductive hypothesis. 
Let $P_l$ be the symbolic matrix of the $l$th anti-diagonal block.
It is sufficient to examine $P_l\ox P$, where $P$ is the symbolic matrix of $\cA(p,q,d)$. 
By the proof of lemma~\ref{compression}, there exist two 
invertible matrices $Q_{l}^1$ and $Q_{l}^2$, such that 
$Q_{l}^1P_l\ox PQ_{l}^2=\begin{pmatrix}P_0&P_1\\P_2&{\bf 0}\end{pmatrix}$, where $P_1$ is the symbolic matrix of 
$\cA_{l0}=\A(p^{N-h(l)+1}(d-p)^{h(l)},q^{h(l)+1}(d-q)^{N-h(l)},p^{N-h(l)}(d-p)^{h(l)}d,q^{h(l)}(d-q)^{N-h(l)}d)$,
and $P_2$ is the symbolic matrix of
$\cA_{l1}=\A(p^{N-h(l)}(d-p)^{h(l)+1},q^{h(l)+2}(d-q)^{N-h(l)-1},p^{N-h(l)-1}(d-p)^{h(l)+1}d,q^{h(l)+1}(d-q)^{N-h(l)-1}d)$,
where $l0$ and $l1$ denote the $N$-bit strings in which the first $(N-1)$-bit 
strings equal the binary expansion of $l$. Moreover, we observe that $\cA_{l0}$ and $\cA_{l1}$ remain as the anti-diagonal blocks in $\cA_1'=\overline{Q_{l}^1}\cA_1\overline{Q_{l}^2}$ ($\overline{Q_{l}^i}$ is the enlarged matrix of $Q_{l}^i$ for $i=1,2$). Then the first fact in observation~\ref{claim:decompose matrix space} follows since $h(l0)=h(l)$ and $h(l1)=h(l)+1$. 

For the second fact, for given $u,v\in\{0,\dots,2^{N-1}-1\}$, $u\neq v$ and $h(u)<h(v)$, we examine $\cA_{u,v}\otimes \cA(p,q,d)$. By induction hypothesis, 

$$\cA_{u,v}=\cA(p^{N-h(u)}(d-p)^{h(u)},q^{h(v)+1}(d-q)^{N-h(v)-1},p^{N-h(u)-1}(d-p)^{h(u)}d,q^{h(v)}(d-q)^{N-h(v)-1}d).$$

Notice that $\cA_{u,v}$ has the same ``full'' rows as that of $\cA_{u}$ and has the same ``full'' columns as that of $\cA_{v}$. Here a ``full'' row (resp. 
column) means the corresponding row (column) of the symbolic matrix contains $*$ only. Denote the row rearrangements of $\cA_u\otimes 
\cA(p,q,d)$ by $R_u$ and the column rearrangements of $\cA_v\otimes \cA(p,q,d)$ by 
$C_v$. These two operations will also rearrange the 
rows and columns of the symbolic matrix of  
$\cA_{u,v}\otimes \cA(p,q,d)$, respectively. For simplicity, we assume 
$\cA_{u}=\cA(p_1,q_1,m_1,n_1)$ and $\cA_{v}=(p_2,q_2,m_2,n_2)$, then 
$\cA_{u,v}=\cA(p_1,q_2,m_1,n_2)$. Let $P_{u,v}$ be the symbolic matrix of $\cA_{u,v}\otimes \cA(p,q,d)$, 
$P_{u,v}$ has the block matrix form
\begin{equation}
P_{u,v}=
\begin{pmatrix}
A_{1,1} & \cdots & A_{1,q_2} & A_{1,q_2+1} & \cdots  & A_{1,n_2} \\
\vdots & P_0 & \vdots & \vdots & P_1  &  \vdots \\
A_{p_1,1} & \cdots & A_{p_1,q_2} & A_{p_1,q_2+1} & \cdots  & A_{p_1,n_2} \\
A_{p_1+1,1} & \cdots & A_{p_1+1,q_2} &  & & & \\

\vdots & P_2 & \vdots &  & \textbf{\huge{0}}& \\
A_{m_1,1} & \cdots & A_{m_1,q_2} &  & & \\
\end{pmatrix},
\end{equation}
where $A_{i,j}$ are symbolic matrix of $\cA(p,q,d)$ for all possible $(i,j)$. Then $R_u$ moves all rows with more than $(n_2-q_2)q$ $*$s in $P_1$ and all rows with more than $q_2q$ $*$s in $P_2$ to the top of them. To see this, notice that all rows with more than $(n_2-q_2)q$ $*$s in $P_1$ is determined by those ``full'' rows in $\cA_{u,v}$, which are exact those rows in $\cA_{u}$; all those rows with more than $q_2q$ $*$s in $P_2$ is determined by those ``full'' rows in $\cA(p,q,d)$, and are also those rows in $\cA_{u}$. Similarly, $C_v$ moves all columns with more than $p_1p$ $*$s in $P_1$ and all columns with more than $(m_1-p_1)p$ $*$s in $P_2$ to the left. Let $P_{u,v}'=\begin{pmatrix}P'_0&P'_1\\P'_2&{\bf 0}\end{pmatrix}$ denotes the symbolic matrix after the row and column rearrangement of $P_{u,v}$. We can then conclude that $P'_0$ is the symbolic matrix of 
$\cA_{0}^*=\cA(p^{N-h(u)+1}(d-p)^{h(u)},q^{h(v)+2}(d-q)^{N-h(v)-1},p^{N-h(u)}(d-p)^{h(u)}d,q^{h(v)+1}(d-q)^{N-h(v)-1}d),$
$P'_1$ is the symbolic matrix of 
$\cA_{1}^*=\cA(p^{N-h(u)+1}(d-p)^{h(u)},q^{h(v)+1}(d-q)^{N-h(v)},p^{N-h(u)}(d-p)^{h(u)}d,q^{h(v)}(d-q)^{N-h(v)}d)$
$P'_2$ is the symbolic matrix of 
$\cA_{2}^*=\cA(p^{N-h(u)}(d-p)^{h(u)+1},q^{h(v)+2}(d-q)^{N-h(v)-1},p^{N-h(u)-1}(d-p)^{h(u)+1}d,q^{h(v)+1}(d-q)^{N-h(v)-1}d).$
Also, $\cA_{0}^*$ will be relabeled as $\cA_{u0, v1}$, $\cA_{1}^*$ will be relabeled as $\cA_{u0,v0}$ and $\cA_{2}^*$ will be relabeled as $\cA_{u1,v1}$ according to their corresponding anti-block terms.

To see the second statement in observation~\ref{claim:decompose matrix space} holds for 
$N+1$, we only need to show that if $h(u)\geq h(v)$, $P_{u,v}={\bf 0}$ ($\cA_{u,v}=\{\bf 0\}$). 
For $u,v\in\{0,\dots,2^{N}-1\}$ and $u\neq v$, let $u=u'b$ and $v=v'c$, where 
$u',v'\in\{0,\dots,2^{N}-1\}$ equal the first $(N-1)$-bit strings of the binary 
expansion of $u$ and $v$, and $b,c\in\{0,1\}$ are variables. If $h(u)\geq h(v)$ derives that either $h(u'0)\geq h(v'0)$, 
$h(u'0)\geq h(v'1)=h(v')+1$ or $h(u'1)\geq h(v'1)$, it will imply that $h(u')\geq 
h(v')$. By the induction hypothesis, $P_{u',v'}={\bf 0}$ and $P_{u,v}={\bf 0}$ holds clearly. Otherwise, if $h(u)\geq h(v)$ derives
$h(u'1)\geq h(v'0)$ and $P_{u',v'}$ is nonzero, we can also observe that $P_{u'1,v'0}$ equals ${\bf 0}$, as it is the lower right part of 
$P^*$. This concludes the proof.
\end{proof}

Now we focus on those anti-diagonal forms, and compute the maximal rank of $\cA_l$ for $l=0,\dots,2^{N-1}-1$:
\begin{observation}\label{obs: rank of anti-block}
Let $h(l)=k$, the maximal rank of $\A_l=\cA(p^{N-k+1}(d-p)^{k},q^{k+1}(d-q)^{N-k},p^{N-k}(d-p)^kd,q^{k}(d-q)^{N-k}d)$ equals
\begin{equation}\label{eq: rank of A p,q,m,n}
\begin{split}
\min\{p^{N-k+1}(d-p)^{k}, q^{k}(d-q)^{N-k+1}\}+\min\{q^{k+1}(d-q)^{N-k},p^{N-k}(d-p)^{k+1}\}.
\end{split}
\end{equation}
\end{observation}
\begin{proof}[Proof of observation~\ref{obs: rank of anti-block}]
Notice that the rank of $\cA(p^{N-k+1}(d-p)^{k},q^{k+1}(d-q)^{N-k},p^{N-k}(d-p)^kd,q^{k}(d-q)^{N-k}d)$ equals
$$\min\{p^{N-k+1}(d-p)^{k}+q^{k+1}(d-q)^{N-k},p^{N-k}(d-p)^kd,q^{k}(d-q)^{N-k}d\}.$$

If $p^{N-k}(d-p)^k\leq q^{k}(d-q)^{N-k}$, we only need to compare 
$p^{N-k+1}(d-p)^{k}+q^{k+1}(d-q)^{N-k}$ and $p^{N-k}(d-p)^kd$. Note that
$$p^{N-k}(d-p)^kd-(p^{N-k+1}(d-p)^{k}+q^{k+1}(d-q)^{N-k})
=p^{N-k}(d-p)^{k+1}-q^{k+1}(d-q)^{N-k}.$$
We further distinguish two cases. When $p^{N-k}(d-p)^{k+1}\geq 
q^{k+1}(d-q)^{N-k}$, we take $p^{N-k+1}(d-p)^{k}+q^{k+1}(d-q)^{N-k}$. When 
$p^{N-k}(d-p)^{k+1}< q^{k+1}(d-q)^{N-k}$, we take 
$p^{N-k}(d-p)^kd=p^{N-k+1}(d-p)^{k}+p^{N-k}(d-p)^{k+1}$. These two cases then can 
be unified in the following equation
\begin{equation}\label{eq:rank 1}
\begin{split}
&\min\{p^{N-k+1}(d-p)^{k}+q^{k+1}(d-q)^{N-k},p^{N-k}(d-p)^kd,q^{k}(d-q)^{N-k}d\}\\
=&p^{N-k+1}(d-p)^{k}+\min\{q^{k+1}(d-q)^{N-k},p^{N-k}(d-p)^{k+1}\}.
\end{split}
\end{equation}

If $p^{N-k}(d-p)^k> q^{k}(d-q)^{N-k}$, similarly, we obtain
\begin{equation}\label{eq:rank 2}
\begin{split}
&\min\{p^{N-k+1}(d-p)^{k}+q^{k+1}(d-q)^{N-k},p^{N-k}(d-p)^kd,q^{k}(d-q)^{N-k}d\}\\
=&\min\{p^{N-k+1}(d-p)^{k}, q^{k}(d-q)^{N-k+1}\}+q^{k+1}(d-q)^{N-k}.
\end{split}
\end{equation}

Notice that, $p^{N-k}(d-p)^k\leq q^{k}(d-q)^{N-k}$ implies $p^{N-k+1}(d-p)^{k}\leq 
pq^{k}(d-q)^{N-k}<q^{k}(d-q)^{N-k+1}$, where the second inequality uses $p+q<d$, 
since $\cA(p,q,d)$ is maximal-compression. Similarly, 
$p^{N-k}(d-p)^k>q^{k}(d-q)^{N-k}$ implies 
$q^{k+1}(d-q)^{N-k}<qp^{N-k}(d-p)^k<p^{N-k}(d-p)^{k+1}$. This observation 
allows us to combine 
equation~(\ref{eq:rank 1}) and equation~(\ref{eq:rank 2}) to obtain a unified 
equation for the maximal rank of 
$\A(p^{N-k+1}(d-p)^{k},q^{k+1}(d-q)^{N-k},p^{N-k}(d-p)^kd,q^{k}(d-q)^{N-k}d)$ 
as 
\begin{equation}\label{eq: rank of A p,q,m,n}
\begin{split}
\min\{p^{N-k+1}(d-p)^{k}, q^{k}(d-q)^{N-k+1}\}+\min\{q^{k+1}(d-q)^{N-k},p^{N-k}(d-p)^{k+1}\}.
\end{split}
\end{equation}
\end{proof}

Finally, we combine observations~\ref{claim:decompose matrix space} and~\ref{obs: rank of anti-block} to prove that $mrk(\cA(p,q,d)^{\ox N+1})=\sum_{l=0}^{2^{N}-1} mrk(\cA_l)$. Then equation~\ref{rankformula} follows. Let $\cA'$ be the matrix space which is obtained after applying row and column rearrangements described in observation~\ref{claim:decompose matrix space} to $\cA(p,q,d)^{\ox N}$. Choose a matrix $P\in\cA'$ of the form as shown in 
equation~(\ref{eq:anti-diagonal matrix form}) with $rank(P)=mrk(\cA')$, we can assume
$rank(P_l)=mrk(\cA_l)$ for $0\leq l\leq 2^N-1$, using an analogous 
argument as in lemma~\ref{compression}. Let $\lambda={\rm log}_2\frac{d-p}{q}$, 
$\mu={\rm log}_2\frac{d-q}{p}$, 
$\alpha=\frac{\mu}{\lambda+\mu}$. Notice that  
$$k\leq \lfloor \alpha N+\alpha-1\rfloor\Leftrightarrow p^{N-k}(d-p)^{k+1}\leq q^{k+1}(d-q)^{N-k}$$ 
and 
$$k\leq \lfloor \alpha N+\alpha\rfloor\Leftrightarrow p^{N-k+1}(d-p)^k\leq q^k(d-q)^{N-k+1}.$$ 
Let $N'=\lfloor \alpha N+\alpha\rfloor=\lfloor \alpha N+\alpha-1\rfloor+1$. For any $l\in\{l:h(l)\leq N'-1\}$, $P_l$ has full row rank. 
For any $l\in\{l:h(l)\geq N'+1\}$, $P_l$ has full column rank.  
Now we claim that, for any upper-anti-block-diagonal matrices $P_{u,v}$, where $u\neq v$ 
and $u,v\in\{2^N-1\}$, there exist row and column operations which convert $P$ 
into the matrix which only has anti-diagonal blocks. By observation~\ref{claim:decompose 
matrix space}, we only need to consider those $P_{u,v}$ satisfying $h(u) < h(v)$. 
In this case, either $h(u)\leq N'-1$, or $h(v)\geq N'+1$. If $h(u)\leq N'-1$, we 
can use $P_u$ to clear $P_{u,v}$, since $P_u$ is anti-diagonal, on the right of 
$P_{u,v}$, and $P_u$ has full row rank. The other case is similar.  These yield 
that $mrk(\A(p,q,d)^{\otimes N+1})=\sum_{l=0}^{2^{N}-1} mrk(\cA_l)$, which, 
together with equation~(\ref{eq: rank of A p,q,m,n}), allow us to conclude the 
proof.
\end{proof}

Now we are ready to compute the asymptotic maximal rank for maximal-compression 
matrix spaces. We restate theorem~\ref{asymptotic} (2) here:

\noindent{\bf Theorem~\ref{asymptotic}. 2, restated.}
If $\cS=\A(p,q,d)$ satisfying $p+q<d$, then 
\begin{equation}\label{eq:asymptotic formula}
mrk^{\infty}(\cS)= d\max\{2^{-D(1-\alpha || p')}, 2^{-D(\alpha || q')}\},
\end{equation}
where $p'=\frac{p}{d}$, $q'=\frac{q}{d}$, $\alpha=\frac{\log_2(d-q)-\log_2 p}{\log_2((d-p)(d-q))-\log_2 (pq)}$ and $D(a||b):=a{\rm log}_2\frac{a}{b}+(1-a){\rm log}_2\frac{1-a}{1-b}$. 

\noindent{\it Proof:} 
Let $\lambda={\rm log}_2\frac{d-p}{q}$, $\mu={\rm log}_2\frac{d-q}{p}$, 
$\alpha=\frac{\log_2(d-q)-\log_2 p}{\log_2((d-p)(d-q))-\log_2 (pq)}=\frac{\mu}{\lambda+\mu}$ and  $N'=\lfloor \alpha N+\alpha\rfloor$ as discussed in lemma~\ref{lemma:rank formula}. We can rewrite 
equation~(\ref{rankformula}) explicitly as the following:
\begin{equation}\label{inequality}
\begin{split}
\rk(\A(p,q,d)^{\otimes (N+1)})&=\sum_{k=0}^{N'-1} \binom{N}{k} p^{N-k}(d-p)^kd+\sum_{k=N'+1}^{N} \binom{N}{k} q^k(d-q)^{N-k}d\\
&+ \binom{N}{N'} \big(p^{N-N'+1}(d-p)^{N'}+q^{N'+1}(d-q)^{N-N'}\big)\\
&=\sum_{k=0}^{N'}\binom{N}{k} p^{N-k}(d-p)^kd+\sum_{k=N'}^{N} \binom{N}{k} q^{k}(d-q)^{N-k}d\\
&-\binom{N}{N'} \big(p^{N-N'}(d-p)^{N'+1}+q^{N'}(d-q)^{N-N'+1}\big).\\
\end{split}
\end{equation}

Let $p'=\frac{p}{d}$ and $q'=\frac{q}{d}$, we have $p'+q'<1$. The above quantity is upper and lower bounded by
\begin{equation}\label{eq:upper bound}
\rk(\A(p,q,d)^{\otimes (N+1)})\leq d^{N+1}\big(\sum_{k=0}^{N'}\binom{N}{k} {p'}^{N-k}(1-p')^k+\sum_{k=0}^{N-N'} \binom{N}{k} {q'}^{N-k}(1-q')^{k}\big);
\end{equation}
\begin{equation}\label{eq:lower bound}
\rk(\A(p,q,d)^{\otimes (N+1)})\geq d^{N+1}\big(\sum_{k=0}^{N'-1}\binom{N}{k} {p'}^{N-k}(1-p')^k+\sum_{k=0}^{N-N'-1} \binom{N}{k} {q'}^{N-k}(1-q')^{k}\big).
\end{equation}

We shall use the following inequalities:
\begin{lemma}[Lemma 4.7.2 in Ref.~\cite{ash1990information}]\label{upperbound}
For $ N' < Np$, we have:
\begin{equation}
\frac{1}{\sqrt{2N}} 2^{-ND(\frac{N'}{N}||p)}\leq \sum_{k=0}^{N'}\binom{N}{k}p^k(1-p)^{N-k}\leq 2^{-ND(\frac{N'}{N} || p)}.
\end{equation} 
\end{lemma}

To apply lemma~\ref{upperbound} to prove equation~(\ref{eq:asymptotic formula}), we need $Nq' < N' < N(1-p')$ holds for sufficiently large $N$. We first prove the following:
\begin{lemma}\label{lemma:q'<alpha<1-p'}
Let $p'$, $q'$ and $\alpha$ be defined as above. $q'<\alpha<1-p'$.
\end{lemma}
\begin{proof}
By expressing $\alpha$ explicitly in terms of $p'$ and $q'$, we need to prove
\begin{equation}\label{equa}
q'<\frac{{\rm log}_2\frac{1-q'}{p'}}{{\rm log}_2\frac{1-q'}{p'}+{\rm log}_2\frac{1-p'}{q'}},~~
1-p'>\frac{{\rm log}_2\frac{1-q'}{p'}}{{\rm log}_2\frac{1-q'}{p'}+{\rm log}_2\frac{1-p'}{q'}}.
\end{equation}
This is equivalent to show 
\begin{equation}\label{equa1}
(1-q')^{1-q'}{q'}^{q'}>{p'}^{1-q'}(1-p')^{q'},~~ 
(1-p')^{1-p'}{p'}^{p'}>(1-q')^{p'}{q'}^{1-p'}.
\end{equation}

Consider the function $f(x,y)=x^y(1-x)^{1-y}$ with $x,y\in (0,1)$. The partial derivative in $x$ is 
$$\frac{\partial}{\partial x}f(x,y)=\frac{x^{y-1}(y-x)}{(1-x)^y}.$$
For any fixed $y$, $\max_{x\in(0,1)}f(x,y)=f(y,y)$. Then
inequality~(\ref{equa1}) holds by choosing $x=1-p',y=q'$ and $x=1-q',y=p'$. 
\end{proof}

Recall $N'=\lfloor \alpha N+\alpha\rfloor$. To ensure that  $Nq' < N' < N(1-p')$, it is 
sufficient to satisfy that $\alpha+\frac{\alpha}{N}< 1-p'$ and 
$q'<\alpha-\frac{1-\alpha}{N}$. Since $\alpha$, $p'$, and $q'$ are fixed, these 
can be achieved as long as 
$N>\max\{\frac{\alpha}{1-p'-\alpha},\frac{1-\alpha}{\alpha-q'}\}>0$.

Applying the upper bound in lemma~\ref{upperbound} 
to inequality~(\ref{eq:upper bound}), we obtain
\begin{equation}
\rk(\A(p,q,d)^{\otimes (N+1)})\leq d^{N+1}(2^{-ND(\frac{N'}{N} || 1-p')}+2^{-ND(1-\frac{N'}{N} || 1-q')}).
\end{equation}
Notice that $-D(a || p)$ is increasing for $0<a<p$, and $\alpha(N+1)-1\leq\lfloor 
\alpha(N+1)\rfloor\leq \alpha(N+1)$. We can replace $\frac{N'}{N}$ by 
$\alpha+\frac{\alpha}{N}$, and $1-\frac{N'}{N}$ by $1-\alpha+\frac{1-\alpha}{N}$,
which gives that
\begin{equation}
\rk(\A(p,q,d)^{\otimes (N+1)})\leq d^{N+1}(2^{-ND(\alpha+\frac{\alpha}{N} || 1-p')}+2^{-ND(1-\alpha+\frac{1-\alpha}{N} || 1-q')}).
\end{equation}
Let $N$ go to infinity. Since $L^p$ norm converges $L^\infty$ norm when $p\to+\infty$, we have
\begin{equation}\label{great}
\rk^\infty(\A(p,q,d))\leq d\max\{{2^{-D(\alpha || 1-p')}},2^{-D(1-\alpha || 1-q')}\}.
\end{equation}

Similarly, applying the lower bound in lemma~\ref{upperbound} to inequality~(\ref{eq:lower bound}), we obtain

\begin{equation}
\begin{split}
\rk(\A(p,q,d)^{\otimes (N+1)})&\geq \frac{d^{N+1}}{(N+1)^2}(2^{-ND(\frac{N'-1}{N} || 1-p')}+2^{-ND(1-\frac{N'-1}{N} || 1-q')})\\
&\geq \frac{d^{N+1}}{(N+1)^2}(2^{-ND(\alpha+\frac{\alpha-2}{N}  || 1-p')}+2^{-ND(1-\alpha+\frac{1-\alpha}{N} ||1-q')}).
\end{split}
\end{equation}
The second inequality holds since $\frac{N'-1}{N}\geq \alpha+\frac{\alpha-2}{N}$ and $1-\frac{N'-1}{N}\geq 1-\alpha+\frac{1-\alpha}{N}$. Thus we have 
\begin{equation}\label{less}
\rk^{\infty}(\A(p,q,d))\geq d\max\{2^{-D(\alpha|| 1-p')}, 2^{-D(1-\alpha || 1-q')}\}.
\end{equation}
Since $D(a|| b)=D(1-a|| 1-b)$, combining inequalities~(\ref{great}) and~(\ref{less}), we have 
\begin{equation}\label{rankf}
\rk^{\infty}(\A(p,q,d))=d\max\{2^{-D(1-\alpha || p')}, 2^{-D(\alpha || q')}\}.
\end{equation}\qed
\medskip

Theorem~\ref{asymptotic} provides explicit formulas for 
computing the asymptotic maximal Schmidt rank of those tripartite pure states 
whose corresponding matrix spaces have no shrunk subspace, or are 
maximal-compression. Taking one step further, we now consider whether the 
asymptotic maximal Schmidt rank of a tripartite state is full. The physical 
interpretation of this problem is, whether the maximally entangled state can be obtained 
from the given tripartite state by means of SLOCC asymptotically.
In the following, we show that this problem is equivalent to the non-commutative rank problem (problem~\ref{prob:non-commutative rank problem}):
\begin{theorem}\label{thm:equivalence}
$\ket{{\rm \Psi}_{\rm ABC}}\in\Hm^{\rm A}_d\otimes\Hm^{\rm B}_d\otimes\Hm^{\rm C}_{d'}$ can be transformed to the $d\otimes d$ maximally entangled state in $\Hm^{\rm A}_d\otimes\Hm^{\rm B}_d$ by means of SLOCC asymptotically, if and only if $M({\rm \Psi}_{\rm ABC})$ does not have shrunk subspace.
\end{theorem}

\begin{proof}
Let $\cS=M({\rm \Psi}_{\rm ABC})\leq M(d,\C)$. It is equivalent to prove $mrk^\infty(\cS)=d$ if and only if $\cS$ has no shrunk subspace.

For the necessary part, we have shown $mrk^\infty(\cS)=d$ if $\cS$ has no shrunk subspace, by theorem~\ref{asymptotic}.1.

For the sufficient part, we prove that, any $d\times d$ matrix space $\cS$ which has shrunk 
subspaces admits $mrk^\infty(\cS)<d$. Let $U\leq \C^d$ be a shrunk subspace of $\cS$ satisfying $\dim(U)=d-q$ and $\dim(\cS(U))=p$. With some proper change of bases, $\cS$ is a subspace of the maximal-compression matrix space $\A(p,q,d)$. By theorem~\ref{asymptotic}.2 and lemma~\ref{lemma:q'<alpha<1-p'}, we have 
$$\rk^{\infty}(\A(p,q,d))=d\max\{2^{-D(1-\alpha || p')}, 2^{-D(\alpha || q')}\},$$
and $q'<\alpha<1-p'$. We can derive $mrk^\infty(\A(p,q,d))<d$, which leads to $mrk^\infty(\cS)<d$. 
\end{proof}

In addition, by theorem~\ref{polynomialtime}, there exist deterministic polynomial-time algorithms for determining whether a given matrix space has 
shrunk subspaces or not. Thus, determining whether the asymptotic maximal Schmidt rank of a tripartite state is full is algorithmically effective, i.e.~
\begin{corollary}
Given a tripartite state $\ket{{\rm \Psi}_{\rm ABC}}\in\Hm^{\rm A}_d\otimes\Hm^{\rm B}_d\otimes\Hm^{\rm C}_{d'}$, there exist deterministic polynomial-time algorithms to determine whether $\ket{{\rm \Psi}_{\rm ABC}}$ can be transformed to the $d\otimes d$ maximally entangled state by means of SLOCC with rate $1$, in an asymptotic setting.
\end{corollary}

\section{Conclusions}\label{conclusion}
In this paper, we exhibit novel results on tripartite-to-bipartite SLOCC transformations, in multi-copy and asymptotic settings. 
We first construct a tripartite pure state $\ket{{\rm \Psi}^d_{\rm ABC}}$ satisfying $msrk(\ket{{\rm \Psi}^d_{\rm ABC}}^{\ox 2})>msrk(\ket{{\rm \Psi}^d_{\rm ABC}})^2$, which implies that the maximal Schmidt rank is strictly super-multiplicative. It also illustrates that, although one copy of $\ket{{\rm \Psi}^d_{\rm ABC}}$ cannot be transformed to the bipartite maximally entangled state by SLOCC, one can do so with two copies. We then provide a full characterization for those tripartite states whose maximal Schmidt rank increase on average under tensor product, by considering the structure of their corresponding matrix spaces. In fact, these tripartite states can be viewed as having advantages in tripartite-to-bipartite SLOCC transformations with multiple copies. Interestingly, except for the degenerated case, this phenomenon holds for all tripartite states of which their maximal Schmidt ranks are not full.

In the asymptotic setting, we consider evaluating the tripartite-to-bipartite entanglement transformation rate of a given tripartite pure state and bipartite pure state. Notably, it equals the logarithm of the asymptotic maximal Schmidt rank of the given tripartite state (where the base of the logarithm is the Schmidt rank of the given bipartite state). Nevertheless, the latter is in general difficult to compute due to its super-multiplicativity. To get around this difficulty, we apply certain results related to the structure of matrix space, including the study of matrix semi-invariants, to obtain explicit formulas which compute the asymptotic maximal Schmidt ranks for a large family of tripartite states. Furthermore, we investigate the asymptotic convertibility to the bipartite maximally entangled state, and show its equivalence to the non-commutative rank 
problem, introduced in Ref.~\cite{Fortin2004}. Based on the recent progress on this problem~\cite{Garg2015,Ivanyos2015a}, there exist
deterministic polynomial-time algorithms to decide whether a tripartite state can be transformed to the maximally entangled state by SLOCC, asymptotically.


\begin{acknowledgements}
We are grateful to Andreas Winter for helpful discussions. We thank the anonymous reviewers for careful reading and numerous suggestions, which help to significantly improve the presentation. RD's research was partly supported by the Australian Research Council (Grant Nos. DP120103776 and FT120100449). YQ's research was partly supported by Australian Research Council (Grant No. DE150100720).
\end{acknowledgements}

\bibliographystyle{ieeetr}
\bibliography{ref}   

\end{document}